\documentclass[11pt]{article}
\usepackage{amssymb, soul}
\usepackage[margin=1.4in]{geometry}
\usepackage{graphicx}
\usepackage{xcolor, soul} % soul for highlight
\usepackage{amsmath}
\usepackage{eurosym}

\newtheorem{theorem}{Theorem}

\newtheorem{definition}[theorem]{Definition}
\newtheorem{example}[theorem]{Example}

\newtheorem{proposition}[theorem]{Proposition}
\newtheorem{remark}[theorem]{Remark}

\newenvironment{proof}[1][Proof]{\textbf{#1.} }{\ \rule{0.5em}{0.5em}}

\def\essinf{\operatorname{ess.\!inf}}

\def\seuro{\operatorname{\mbox{\small \euro}}}

\begin{document}

\title{Capturing cash non-additivity and horizon risk\\
via BSDEs and generalized shortfall}
\author{Giulia Di Nunno\thanks{Department of Mathematics,
University of Oslo, P.O. Box 1053 Blindern, 0316 Oslo Norway.
Email: giulian@math.uio.no
and NHH - Norwegian School of Economics, Helleveien 30, 5045 Bergen, Norway.}
\thanks{This research is carried out in the frame of the SURE-AI Center, grant nr. 357482 of the Research Council of Norway.}
  \and Emanuela Rosazza Gianin\thanks{Department of Statistics and Quantitative Methods,
University of Milano-Bicocca, via Bicocca degli Arcimboldi 8, 20126 Milano Italy.
Email:  emanuela.rosazza1@unimib.it. This author is member of GNAMPA-INdAM, Italy.}}
\date{\today}

\maketitle
\vspace{-0.7cm}
\begin{abstract}
Whenever dealing with horizons of different times scales, risk evaluation of losses may incur in both interest rate uncertainty and horizon risk as introduced in \cite{DNRG1}. With the goal to capture both effects, we work with cash subadditive fully-dynamic risk measures.
In this work, we consider such measures obtained via the BSDE and the shortfall approaches. We stress that we consider BSDEs both with Lipschitz and quadratic drivers. We then introduce the hq-entropic risk measure on losses as an effective example of fully-dynamic risk measure serving the scope. 
Shortfall risk measures are extended to capture cash non-additivity. For our newly introduced h-generalized shortfall risk measures we provide a dual representation and we connect them to fully-dynamic certainty equivalent. 
To conclude, we can see that the hq-entropic risk measures on losses belong to the family h-generalized shortfall, but they are not of certainty equivalent type. We note that the classical entropic risk measure, besides being generated by a BSDE, is also both a shortfall and a certainty equivalent.

\vspace{2mm}\noindent

\vspace{2mm}\noindent
\end{abstract}

\section{Introduction}

Horizon risk is defined as the use of a risk measure designed for long term positions when evaluating short term investments, and vice versa.
This risk is particularly critical, in the context of, e.g., pensions and health insurance, where long-term claims are to be expected and hedged.
In this context, for instance, we know that the use of outdated mortality rate, or the choice of an incorrect cohort may lead to wrong premia evaluations with a consequent impact on capital requirements.
In \cite{DNRG1}, horizon risk has been identified using fully-dynamic risk measures and introducing the notion of {\it horizon longevity} or {\it h-longevity}, in short, as an index of quantification.
Indeed, a {\it fully-dynamic risk measure} is a family $(\rho_{tu})_{t,u}\triangleq (\rho_{tu})_{0 \leq t \leq u \leq T}$ of risk measures indexed by two time parameters, which naturally takes care of the horizon in an explicit form, see~\cite{bion-nadal-di-nunno}.
To understand horizon risk, we must recall the {\it restriction property} (see~\cite{bion-nadal-di-nunno}):
\begin{equation}
\label{eq: restriction}
\rho_{tu}(X) = \rho_{tv}(X), \qquad \text{for }  v \geq u, \text{ and } X \; \text{$\mathcal{F}_u$-measurable}, \quad
\end{equation}
with respect to the information flow $(\mathcal{F}_t)_{t\in [0,T]}$.
Roughly speaking, this means that a risky position $X$ at a short horizon $u$ is evaluated as if it was happening at a longer horizon $v$.
On the contrary, whenever the restriction is lifted, we have an open possibility to quantify horizon risk. For this, the notion of {\it h-longevity} has been proposed in~\cite{DNRG1} in terms of a correction term:
  \begin{equation}
  \label{eq: time-value}
  \gamma(t,u,v,X) \triangleq  \rho_{tv}(X) - \rho_{tu} (X) \geq 0,  \quad   t \leq u \leq v.
  \end{equation}

When working with different time-scales of short and long time horizons, we must take the value of money into account via interest rates, which may be uncertain, see e.g. \cite{filipovic, EK-rav} and the discussion in~\cite{farkas}.
The explicit use of interest rates leads directly to the introduction of cash subadditive risk measures as noted in~\cite{EK-rav}.
We recall that {\it cash additivity} (also called {\it $\mathcal{F}_t$-translation invariance}) refers to:
\begin{equation}
\label{eq: CA}
\rho_{tu} (X + m) = \rho_{tu} (X ) - m, 
\end{equation}
for $\mathcal{F}_u$-measurable $X$ and  $\mathcal{F}_t$-measurable $m$.
When \eqref{eq: CA} does not hold, the risk measure is called {\it cash non-additive}. In particular, it is called {\it cash subadditive} (see~\cite{EK-rav} and also~\cite{mastrogiacomo-rg, wang-csa}) when
\begin{equation*}
%\label{eq: CSA}
\rho_{tu} (X + m) \geq \rho_{tu} (X ) - m, \qquad m \geq 0.
\end{equation*}

\vspace{2mm}
To explain the interaction with interest rate uncertainty, we must make explicit the quantities in terms of unit of money. Thus we call $\seuro_u$ the unit of money at time $u$. Hence a financial investment available at $u$ is denoted $X\seuro_u$, where $X$ represents the size of the investment.
Also, let $(D_{tu})_{t,u}$ be the family of discount factors $D_{tu}$ on the time interval $(t,u]$. The unit of measurement for $D_{tu}$ is then $1/ \seuro_u$.
It is then assumed that $0 < d_{tu} \leq D_{tu} \seuro_u \leq 1$, for some (stochastic) lower bound $d_{tu}$.
Then, for any cash additive fully-dynamic risk measure $(\phi_{tu})_{t,u}$, we define
\begin{equation*}
\rho_{tu}(X) \triangleq \phi_{tu} (D_{tu} X \seuro_u) , \qquad \text{for }\; X \; \text{$\mathcal{F}_u$-measurable},
\end{equation*}
which is cash subadditive. In fact, since $D_{tu}\seuro_u \leq 1$, by monotonicity and cash additivity of $\phi_{tu}$, we have
\begin{eqnarray*}
\rho_{tu}(X+m) &=&  \phi_{tu} (D_{tu} (X + m) \seuro_u  ) \\
&\geq & \phi_{tu} (D_{tu} X \seuro_u + m )  = \phi_{tu} (D_{tu} X \seuro_u)- m =  \rho_{tu}(X)-m,
\end{eqnarray*}
for any $X$ $\mathcal{F}_u$-measurable and $m\geq 0$ $\mathcal{F}_t$-measurable.
This observation triggers the interest in cash subadditive and more generally cash non-additive risk measures, which will be discussed in this work, with the goal of providing a family of fully-dynamic risk measures satisfying both horizon longevity and cash non-additivity.

\vspace{2mm}
In our previous work \cite{DNRG1}, we have extended the classical entropic risk measure to its h-entropic version, so to capture horizon risk. Indeed, the \textit{entropic risk measure} $\rho_{tu}(X)$ is generated by a BSDE of type:
\begin{equation*}
-dY_t = \frac{b}{2}\vert Z_t \vert^2 dt - Z_t dB_t, \quad Y_u = -X,
\end{equation*}
with $(B_t)_{t\geq 0}$ d-dimensional Brownian motion and risk aversion $b>0$, leading to
\begin{equation}
\label{eq: entropica-expl}
\rho_{tu}(X)= Y_t=\frac{1}{b}\ln E  \left[\left. \exp(-bX) \right| \mathcal{F}_t \right].
\end{equation}

In \cite{DNRG1}, we have introduced the \textit{h-entropic risk measure} generated by a BSDE with driver
\begin{equation*}
 g^h(t,y,z) = \frac{b}{2}\vert z \vert^2 +a(t) , \quad t \leq u, z\in \mathbb{R^d},
 \end{equation*}
where $a$ is a non-negative real function. The risk measure is then given by
\begin{align}
\label{eq: hentropica-expl}
\rho^h_{tu}(X) =Y_t 
&=
\frac{1}{b}\ln E  \left[\left. \exp\Big(-bX + \int_0^u ba(s) \,ds\Big) \right| \mathcal{F}_t \right] - \int_0^t a(s) \,ds \notag\\
&=
\frac{1}{b}\ln E  \left[\left. \exp\Big(-bX \Big) \right| \mathcal{F}_t \right] + \int_t^u a(s) \,ds \notag \\
&=
\frac{1}{b}\ln E  \left[\left. \exp\Big(-bX+ \int_t^u ba(s) \,ds\Big) \right| \mathcal{F}_t \right].
\end{align}
Note that the last two representations are based on the properties of the functions $\exp$ and $\ln$.
To simplify the notation in the sequel, we set $b=1$ and we use the shorthand writing $A(c,d) \triangleq \int_c^d a(s)ds$. 

At this stage, the two measures \eqref{eq: entropica-expl} and \eqref{eq: hentropica-expl} are both cash additive.
Inspired by the works of \cite{ma-tian, tian}, we aim at defining a class of entropic type risk measures both cash non-additive and able to capture horizon risk.
Specifically, we are interested in the evaluation of losses exceeding a given level of severity $\beta \geq 0$.
For this, we introduce the \textit{hq-entropic risk measure on losses} induced by the BSDE
\begin{equation*}
    -dY_t = \Big[ \frac{q}{2} \, \frac{\vert Z_t \vert ^2}{1+ (1-q)(A(0,t)+ Y_t)} + a(t) \Big] dt - Z_t dB_t, \qquad Y_u = (X+\beta)^- + \alpha_q,  
\end{equation*}
for some $q\in (0,1)$ and $\alpha_q > \frac{1}{q-1}$, representing the risk attitude.
The risk measure has the following explicit formulation
\begin{equation}\label{eq: hq-entropic-intro}
    \rho^{hq}_{tu}(X)= Y_t =\ln_q E  \left[\left. \exp_q\Big((X+\beta)^- +\alpha_q+ A(0,u) \Big) \right| \mathcal{F}_t \right] - A(0,t),
\end{equation}
in terms of the generalized logarithm and exponential functions ($\ln_q$ and $\exp_q$)\footnote{
From~\cite{tsallis1}, recall that $\ln_q(x) \to \ln x$ and $\exp_q (x) \to \exp (x)$, for $q \to 1$, where
\vspace{-2mm}
$$\exp_q(x)\triangleq [1+(1-q)x]^{\frac{1}{1-q}} \,\Bigg\{
\begin{array}{l}
\mbox{for } x\geq \frac{1}{q-1}, q \in (0,1)\\
\mbox{for }x< \frac{1}{q-1}, q>1
\end{array}
; \qquad
\ln_q(x) \triangleq \frac{x^{1-q}-1}{1-q}
\,\Bigg\{
\begin{array}{l}
\mbox{for } x\geq 0, q \in (0,1)\\
\mbox{for }x>0, q>1
\end{array}
$$}.
See Tsallis relative entropy~\cite{tsallis1, tsallis2}.

 The first part of this work is dedicated to formalizing the concepts above. We detail the properties of h-longevity and cash non-additivity in the BSDE approach both with Lipschitz and quadratic drivers, so to accommodate the extensions of entropic risk measures we will introduce.
 
\vspace{2mm}
The classical entropic risk measure can also be regarded as a \textit{shortfall risk measure}, that is
\begin{equation*}
    \rho_{tu} (X) = \essinf \big\{m_t \; \mathcal{F}_t-\mbox{measurable} : \; E \big[ U(X+m_t) \vert \mathcal{F}_t \big] \geq B \big\}
\end{equation*}
for $X$ $\mathcal{F}_u$-measurable, with $U$ utility function and $B$ a fixed target level. 
See \cite{follmer-schied-book}.
Motivated by this connection, we study the extension of shortfall risk measures to capture both cash non-additivity and horizon risk.
This leads to the concept of \emph{h-generalized shortfall risk measure} as 
 \begin{equation*}
    \rho^{U,f,B}_{tu} (X) = \essinf \big\{m_t \; \mathcal{F}_t-\mbox{measurable} : \; E \big[ U_u(f_{tu}(X,m_t)) \vert \mathcal{F}_t \big] \geq B_{tu} \big\}
\end{equation*}
 by means of the families of concave utility functions $(U_u)_u$, real aggregator functions $(f_{tu})_{t,u}$, and $\mathcal{F}_t$-measurable target levels $(B_{tu})_{t,u}$.
 Even further, we have studied the relationship between the h-generalized shortfall and the so-called \emph{certainty equivalent risk measure}, which we present in the fully-dynamic version:
\begin{equation*} 
\tilde{U}_u (-\rho^{\tilde U_u}_{tu}(X))=E[\tilde{U}_u (X) |\mathcal{F}_t], 
\end{equation*}
for $(\tilde{U}_u)_{u }$ of concave functions. 
Indeed our results show that a h-generalized shortfall is a certainty equivalent under specific relationship between parameters $(U_u,f_{tu}, B_{tu})_{t,u}$ and $(\tilde U_u)_u$.

Winding back to our hq-entropic risk measure \eqref{eq: hq-entropic-intro}, we investigate whether it is indeed a h-generalized shortfall and a fully-dynamic certainty equivalent risk measure. 
We see that while it is true that \eqref{eq: hq-entropic-intro} can be represented as a h-generalized shortfall by means of a clever selection of parameters, it is however not true that it is of certainty equivalent type.

\smallskip
Our work is organized as follows.
In Section~\ref{Sec: BSDE}, we present the BSDE approach to the construction of fully-dynamic risk measures. We consider both the Lipschitz  and the quadratic drivers, bearing in mind to work with the entropic risk measures. We focus on cash non-additivity and characterize h-longevity. In the Lipschitz case we achieve an explicit formula for the h-longevity index.
To complement the study we also suggest to construct fully-dynamic risk measure from a family of BSDEs instead of a single one. This approach exacerbates the role of horizon risk.
In Section \ref{Sec3} we introduce and study the hq-entropic risk measure on losses.
Section \ref{Sec5} is proposing the generalized shortfall representation (static and dynamic) to embed both cash non-additivity and h-longevity. These measures are quasi-convex. We study their dual representation. 
In Section \ref{Sec CE} we consider the relationship between h-generalized shortfall and fully-dynamic certainty equivalent risk measures. 
In the last section, we recognize the hq-entropic risk measure on losses as a generalized shortfall. We note, however, that it is not a fully-dynamic certainty equivalent.

 %%%%%%%%%%%%%%%%%%%%

\section{BSDEs' approach}
\label{Sec: BSDE}

On a complete probability space $(\Omega, \mathcal{F},P)$, let $(B_t)_{t \in [0,T]}$ be a $d$-dimensional Brownian motion and $(\mathcal{F}_t)_{t \in [0,T]}$ its $P$-augmented natural filtration. 
%We restrict our attention on $L^2$ spaces.
We focus on fully-dynamic risk measures $(\rho_{tu})_{t,u}$ associated to BSDEs of the following form
\begin{equation} \label{eq: BSDE}
Y_t= X+ \int_t^u g(s,Y_s,Z_s) \, ds - \int_t^u Z_s \, dB_s \quad (u\in [0,T]),
\end{equation}
whose solution $(Y_t,Z_t)_t \triangleq (Y_t,Z_t)_{t \in [0,u]}$ can be seen as a nonlinear operator depending on the driver $g$ and evaluated at the final condition $X \in L^2( \mathcal{F}_u)$. See~\cite{peng97} for the Lipschitz case and \cite{bahlali-et-al} for the linear-quadratic case. About this, see also \cite{bahlali-tangpi} and the seminal paper by \cite{kolylanski} for bounded terminal conditions. In Peng's terminology,
$\mathcal{E}^g \left(X \vert \mathcal{F}_t \right)$ denotes the $Y$-component of the solution $(Y_t,Z_t)$ at time $t$ of the BSDE above, called conditional $g$-expectation of $X$ at time $t$.

The existence of a solution to \eqref{eq: BSDE} depends on the properties of the parameters $(X,g)$.
Typically, the literature calls {\it standard assumptions} $X\in L^2(\mathcal{F}_u)$ and
a progressively measurable driver $g: \Omega \times [0,u]\times \Bbb R \times \Bbb R^d \to \Bbb R$
satisfying:
\begin{itemize}
\item uniformly Lipschitz, i.e. there exists a constant $C>0$ such that, $dP \times dt$-a.e.,
\begin{equation*}
\vert g(\omega,t, y_1, z_1) - g(\omega, t , y_2, z_2)\vert \leq C (\vert y_1 - y_2 \vert + \vert z_1 - z_2\vert ),
\end{equation*}
for any $y_1, y_2 \in \Bbb R, \, z_1, z_2 \in \Bbb R^d$,
where $\vert \cdot \vert$ denotes the Euclidean norm in $\Bbb R^k$;
\item $E\left[\int_0^u \vert g(t,0,0)\vert ^2 \, dt \right]<+\infty$.
\end{itemize}
Such standard assumptions guarantee that equation~\eqref{eq: BSDE} admits a unique solution $(Y_t,Z_t)_{t}$, with $(Y_t)_{t} \in \Bbb H^2_{[0,u]} (\Bbb R)$ and $(Z_t)_{t} \in \Bbb H^2_{[0,u]} (\Bbb R^d)$, where
\begin{equation*}%\label{eq: def H2}
\Bbb H^2_{[a,b]} (\Bbb R^k) \hspace{-1mm} \triangleq \hspace{-1mm} \Big\{
\mbox{adapted } \Bbb R^k\mbox{-valued processes } (\eta_t)_{t \in [a,b]}: \hspace{-1mm} E\Big[\hspace{-1mm}\int_a^b \hspace{-2mm} \vert \eta_t \vert^2 \, dt \Big] \hspace{-1mm} < \hspace{-1mm} \infty
\Big\}.
\end{equation*}
For later use, we also introduce 
\begin{equation*}%\label{eq: def calH2}
\mathcal{H}^2_{[a,b]} (\Bbb R^k) \hspace{-1mm} \triangleq \hspace{-1mm} \Big\{
\mbox{prog. meas. } \Bbb R^k\mbox{-valued processes } (\eta_t)_{t \in [a,b]}: \hspace{-1mm} \int_a^b \hspace{-2mm} \vert \eta_t \vert^2 \, dt  \hspace{-1mm} < \hspace{-1mm} \infty \; P-a.s.
\Big\}.
\end{equation*}

The quadratic case is more delicate. Large part of the literature considers $X\in L^\infty(\mathcal{F}_u)$ and drivers $g$ that are continuous in $(y,z)$ and satisfying
$$
\vert g(t,y,z) \vert \leq a_t + b_t \vert y \vert +c_t \vert z \vert + f(y) \vert z\vert^2
$$
with different assumptions on $a,b,c$ and $f$.
See \cite{kolylanski, bahlali-et-al,bahlali-tangpi}. While the existence of the solution is established, uniqueness is yet another challenge. The context of $X\in L^2(\mathcal{F}_u)$ is analyzed in \cite{bahlali-et-al}, where existence and uniqueness are proved for $g(t,y,z) = f(y) \vert z \vert^2$, with $f$ globally integrable.
Their argument can be easily extended as follows.
\begin{proposition} \label{prop: quadexistence}
For $X\in L^2(\mathcal{F}_u)$ and 
\begin{equation} \label{eq: gquad}
    g(t,y,z) = a_t + f(y) \vert z \vert^2
\end{equation}
with $a\in L^1([0,u])$ and $f\in L^1(\mathbb{R)}$, the BSDE \eqref{eq: BSDE} has a unique solution in $ \Bbb H^2_{[0,u]} (\Bbb R) \times \Bbb H^2_{[0,u]} (\Bbb R^d)$.
Moreover, let $a$ and $f$ be non-negative. For $X \geq 0$, the $Y$-component of the unique solution is non-negative. 
\end{proposition}
\begin{proof}
    Define $\tilde{f}(x) = f(x-\int_0^t a_s ds)$, $x\in \mathbb{R}$, and $\tilde{X} = X + \int_0^u a_s ds$.
    The BSDE
    $$
   \eta_t= \tilde X+ \int_t^u \tilde f(\eta_s)\vert \zeta_s \vert^2 \, ds - \int_t^u \zeta_s \, dB_s 
    $$
    has a unique solution $(\eta_t, \zeta_t)_t$ (see \cite[Prop.~3.1]{bahlali-et-al}).
    We easily see that $Y_t = \eta_t - \int_0^t a_s ds$ and $Z_t = \zeta_t$ satisfy \eqref{eq: BSDE} with \eqref{eq: gquad}.

Assume now that $a$ and $f$ are non-negative. For $X \geq 0$, by an immediate adjustment of the Comparison Theorem in \cite[Prop. 3.2]{bahlali-et-al} to include the term $a$, we see that the solution is non-negative.
\end{proof}

\vspace{2mm}
In the sequel we are dealing with BSDEs \eqref{eq: BSDE}-\eqref{eq: gquad}, with a singularity in $f$ for which global integrability is lost.
Hereafter, we propose a self-contained proof of the existence and uniqueness of the solution of the BSDE \eqref{eq: BSDE} with driver of type \eqref{eq: BSDEsing0} and final condition $X>0$ satisfying some integrability condition.
In the literature, we find studies dealing with the case where $\kappa\equiv 0$ within the works \cite{bahlali-tangpi} and  \cite{laeven-et-al2}. 

\begin{theorem}
\label{th: sing0}
Let $\delta, \epsilon >0$ and $\kappa\in \mathbb{R}$. Consider the BSDE \eqref{eq: BSDE} with driver 
\begin{equation}
\label{eq: BSDEsing0}
g(t,y,z) =  \frac{\delta}{\kappa + y} \vert z \vert^2, \quad y \geq - \kappa + \epsilon > - \kappa.
\end{equation}
Define  $\mathfrak{D}_u$ the set of all $\mathcal{F}_u$-measurable random variable $X$ with $X \geq -\kappa + \epsilon$ and such that 
\begin{equation}
\label{eq: intcond}
 E\big[ (\kappa + X)^{2\delta +1} \ln^+(\kappa + X) ] < \infty,
 \end{equation}
where $\ln^+ x \triangleq \max\{0, \ln x\}$.
Then the BSDE has a unique solution $(Y,Z)$ with $Z \in \mathcal{H}^2_{[0,u]}(\mathbb{R}^d)$ and $Y$ such that
$$
E\big[  \sup_{t\in [0,u]} \vert Y_t \vert \big] <\infty .
$$
Furthermore, the process $Y$ has explicit representation as
\begin{equation*}
\label{eq: BSDErep}
Y_t = \ln_q E\Big[ \exp_q \Big(\frac{1}{q-1} + \kappa + X\Big) \Big\vert \mathcal{F}_t\Big] + \frac{1}{1-q} -\kappa ,
\end{equation*}
with $q \triangleq \frac{2\delta}{2\delta +1} \in (0,1)$.
\end{theorem}
\begin{proof}
Define the function
$$
\Theta(y) = \frac{1}{2\delta + 1} \big( \kappa + y\big)^{2\delta +1}, \quad y >-\kappa,
$$
which is a one-to-one function, twice differentiable with continuity. Hence, 
\[
\begin{split}
d\Theta( Y_t) &= (\kappa + Y_t)^{2\delta} d Y_t + \delta (\kappa +  Y_t)^{2\delta -1} d <Y>_t  \\
&= (\kappa +  Y_t)^{2\delta}  Z_t dB_t.
\end{split}
\]
So, the couple $\hat Y_t \triangleq \Theta( Y_t)$ and $\hat Z_t\triangleq (\kappa + Y_t)^{2\delta}  Z_t $, formally satisfies the linear BSDE
\begin{equation}
\label{eq: BSDElin0}
 \hat Y_t = \hat X - \int_t^u \hat Z_s dB_s
\end{equation}
with $\hat X \triangleq \hat Y_u = \frac{1}{2\delta + 1} (\kappa + X)^{2\delta + 1}$ and the converse is also true. 
Namely, the couple 
$$
 Y_t = \Theta^{-1}(\hat Y_t) = \big( ( 2\delta +1) \hat Y_t\big)^\frac{1}{2\delta +1} - \kappa
\qquad
 Z_t = \big( ( 2\delta +1) \hat Y_t\big)^{-\frac{2\delta}{2\delta +1}} \hat Z_t
$$
satisfies the BSDE with driver \eqref{eq: BSDEsing0}. This shows a one-to-one correspondence between the two solutions. 
The unique solution $(\hat Y, \hat Z)$ of the BSDE \eqref{eq: BSDElin0} is provided by the local martingale representation theorem.
So, there exists a unique $\hat Z \in \mathcal{H}^2_{[0,u]}(\mathbb{R}^d)$ and
$$
\hat Y_t = E\big[ \hat X \vert \mathcal{F}_t] = \hat X - \int_t^u \hat Z_s dB_s,
$$
which is actually a positive martingale.
Indeed, $\hat X = \frac{1}{2\delta + 1} (\kappa + X)^{2\delta + 1} \in L^1(P)$, since \eqref{eq: intcond} holds:
\begin{equation}
    \label{eq: estimate}
E\big[(\kappa + X)^{2\delta + 1}\big]  \leq   e^{2\delta + 1} + (2\delta +1)E\big[(\kappa + X)^{2\delta + 1} \ln^+(\kappa + X)\big] <\infty.
\end{equation}
Here we have used $x^\alpha \leq e^\alpha + \alpha x^\alpha \ln^+x$, $x>0$, $\alpha >0$.
Observe that $\hat Y_t >0$ for all $t$.
From the Doob $L\ln L$ maximal inequality (see \cite{gilat}) and \eqref{eq: intcond}, we also have
\[\begin{split}
E\big[ \sup_{t\in [0,u] } \hat Y_t \big] 
&\leq \frac{e}{e-1} \big( 1 + E \big[ \hat X \ln^+ (\hat X) \big] \Big)\\
&=  \frac{e}{e-1}\Big( 1 +(2\delta + 1)E \big[ (\kappa + X)^{2\delta + 1} \ln^+ (\kappa + X) \big] \Big)<\infty.
\end{split}\]
Correspondingly, for the processes $(Y, Z)$, we can see that, $Y_t > - \kappa $ for all $t$, and 
\[\begin{split}
E\big[ \sup_{t\in [0,u] }  \vert  Y_t \vert \big]
\leq  
A_\delta  E\big[ \sup_{t\in [0,u] } \hat Y_t\big] +B_{\delta,\kappa}  <\infty,
\end{split}\]
for some positive constants $A_\delta, B_{\delta,\kappa}$, thanks to the concavity $x^\alpha \leq \alpha x + (1-\alpha)$, $x>0$, $\alpha\in (0,1)$.
Hence,
$$
-\kappa <-\kappa +\epsilon \leq \inf_{t\in [0,u] }  Y_t < \sup_{t\in [0,u] }   Y_t < \infty \quad P-a.s.
$$
Then 
$$
\int_0^u \vert  Z_t \vert^2 dt 
= \int_0^u \vert (\kappa + Y_t)^{-2\delta} \hat Z_t \vert^2 dt
\leq (\kappa + \inf_{t\in [0,u] }   Y_t)^{-4\delta} \int_0^u \vert \hat Z_t \vert^2 dt <\infty \quad P-a.s.
$$
We conclude that we have a well defined unique solution of the BSDE with driver \eqref{eq: BSDEsing0}.
Set now $q \triangleq \frac{2\delta}{2\delta +1} \in (0,1)$, and consider the (local) martingale 
$$
\mathcal{Y}_t  \triangleq  E  \Big[ \exp_q\Big(\frac{1}{q-1}+\kappa + X\Big) \Big| \mathcal{F}_t \Big]= \exp_q\Big(\frac{1}{q-1}+\kappa + X\Big) - \int_t^u \mathcal{Z}_s dB_s.
$$
Note that $\exp_q(\frac{1}{q-1}+\kappa + X)$ is well defined for $X>-\kappa$, and $E\big[ \exp_q(\frac{1}{q-1}+\kappa + X)\big] <\infty$ if and only if 
$E\big[ (\kappa + X)^{2\delta + 1}\big] <\infty$, which is guaranteed by \eqref{eq: intcond}.
We apply It\^o formula to find that
$$
-d\Big(\ln_q \mathcal{Y}_t -\kappa + \frac{1}{1-q}\Big)
= \frac{q}{2} \: \mathcal{Y}_t^{q-1} \: (\mathcal{Y}_t^{-q} \mathcal{Z}_t)^2 dt -  \mathcal{Y}_t^{-q} \mathcal{Z}_t dB_t.
$$
Observe that
$$
\mathcal{Y}_t^{q-1} 
= \big( \exp_q \ln_q \mathcal{Y}_t \big)^{q-1}
= \frac{1}{(1-q) \big[ \kappa +\big(\ln_q  \mathcal{Y}_t  -\kappa + \frac{1}{1-q}\big)\big]}
$$
Taking $ Y_t = \ln_q \mathcal{Y}_t -\kappa + \frac{1}{1-q}$ and $Z_t = \mathcal{Y}_t^{-q} \mathcal{Z}_t$, we see that the dynamics above corresponds to
$$
-dY_t = \delta \:\frac{\vert Z_t\vert^2}{\kappa+ Y_t} \:dt - Z_t dB_t, \qquad \delta= \frac{q}{2(1-q)}.
$$
Hence the solution $(Y,Z)$ to the BSDE with driver \eqref{eq: BSDEsing0} has representation
$$
Y_t = \ln_q E\Big[ \exp_q \Big(\frac{1}{q-1} + \kappa + X\Big) \Big\vert \mathcal{F}_t\Big] + \frac{1}{1-q}-\kappa . 
$$
\end{proof}

\begin{theorem}
\label{th: singa}
For $\epsilon >0$, $\kappa \in \mathbb{R}$ and $A(0,t) \triangleq \int_0^t a(s) ds$, with non-negative $a\in L^1([0,u])$, let $\mathfrak{D}^a_{qu}$
be the set of $\mathcal{F}_u$-random variables $X$ such that $X\geq - \kappa +\epsilon$ and
\begin{equation}
\label{eq: intcond2}
E\big[ (\kappa +A(0,u) + X)^{\frac{1}{1-q}} \ln^+(\kappa + A(0,u)+ X) ] < \infty,
\end{equation}
for some $q\in (0,1)$.
Then, for any $X\in \mathfrak{D}^a_{qu}$, the stochastic process
\begin{equation}
    \label{eq: solBSDEsing1}
Y_t = \ln_q E\Big[ \exp_q \Big(\frac{1}{q-1} + \kappa + X +A(0,u)\Big) \Big\vert \mathcal{F}_t\Big] + \frac{1}{1-q}-\kappa - A(0,t),
\end{equation}
satisfies
$$
E\big[\sup_{t\in [0,u]} \vert Y_t \vert \big] <\infty
$$
and there exists a process $Z \in  \mathcal{H}^2_{[0,u]}(\mathbb{R}^d)$ such that the couple $(Y,Z)$ solves the BSDE
\begin{equation}
    \label{eq: BSDEsing}
    Y_t = X + \int_t^u \Big[  \frac{\delta \vert Z_t\vert^2}{\kappa + A(0,s) + Y_s} +a(s) \Big] dt - \int_t^u Z_s dB_s, \qquad \delta = \frac{q}{2(1-q)}.
\end{equation}
\end{theorem}
\begin{proof}
Observe that $\exp_q \Big(\frac{1}{q-1} + \kappa + X +A(0,u)\Big) \in L^1(\mathcal{F}_u)$ if and only if $X^\frac{1}{1-q}\in L^1(\mathcal{F}_u)$, which is guaranteed by \eqref{eq: intcond2}.
From \eqref{eq: solBSDEsing1}, we see that
\begin{align*}
    \mathcal{Y}_t &\triangleq  \exp_q \Big( Y_t + \frac{1}{q-1} + \kappa +A(0,t) \Big) \\
    &=
    E\Big[ \exp_q \Big(\frac{1}{q-1} + \kappa + X +A(0,u)\Big) \Big\vert \mathcal{F}_t\Big]\\
    &= E\Big[ \exp_q \Big(\frac{1}{q-1} + \kappa + X +A(0,u)\Big) \Big] + \int_0^t \mathcal{Z}_s dB_s,
\end{align*}
for a unique $\mathcal{Z}\in \mathcal{H}^2_{[0,u]}(\mathbb{R}^s)$, thanks to the (local) martingale representation theorem.
Note that for every $t$, $\mathcal{Y}_t>0$, since $X\in \mathfrak{D}^a_{qu}$. 
From $Y_t = \ln_q \mathcal{Y}_t -A(0,t) - \frac{1}{q-1} -\kappa$, we have
\begin{align*}
dY_t &= \mathcal{Y}_t^{-q} d\mathcal{Y}_t -  \frac{q}{2} \mathcal{Y}_t^{-q-1} d <\mathcal{Y}>_t - a(t) dt \\
&= - \Big[  \frac{q}{2}  \vert \mathcal{Y}_t^{-q} \mathcal{Z}_t \vert^2\mathcal{Y}_t^{q-1} + a(t)  \Big] dt + \big(\mathcal{Y}_t^{-q} \mathcal{Z}_t \big) dB_t .
\end{align*}
Observe that
$$
\mathcal{Y}_t^{q-1} = \Big[  \exp_q \Big( Y_t + \frac{1}{q-1} + \kappa +A(0,t) \Big)\Big]^{q-1} 
= \frac{1}{(1-q) (\kappa + A(0,t) + Y_t)}.
$$
Hence, taking $Z_t = \mathcal{Y}_t^{-q}\mathcal{Z}_t$, we see that
$$
- dY_t = \Big[ \frac{\delta \vert Z_t \vert^2}{\kappa+A(0,t) + Y_t} +a(t)  \Big]dt - Z_tdB_t, \qquad \delta = \frac{q}{2(1-q)},
$$
proving \eqref{eq: BSDEsing}.
It is easy to see that
$$
E\big[  \sup_{t\in [0,u]} \vert Y_t \vert \big] <\infty \quad \Longleftrightarrow \quad E\big[  \sup_{t\in [0,u]} \vert \ln_q \mathcal{Y}_t \vert \big] <\infty.
$$
Hence, we study the right-hand side. 
The inequality $\vert \ln_q x \vert \leq \frac{1}{1-q} +1 +\vert x \vert$, $x\geq 0$, holds true and  
$$
E\big[  \sup_{t\in [0,u]} \vert \ln_q \mathcal{Y}_t \vert \big] <\infty \quad \Longleftrightarrow \quad 
E\big[  \sup_{t\in [0,u]}  \mathcal{Y}_t  \big] <\infty .
$$
The Doob $L\ln L$ maximal inequality provides the estimate
\begin{align*}
E\big[  \sup_{t\in [0,u]}  \mathcal{Y}_t  \big] 
&\leq 
\frac{e}{e-1} \big( 1 + E\big[ \mathcal{Y}_u \ln^+ \mathcal{Y}_u \big] \big) \\
&=
\frac{e}{e-1} \Big( 1+ 
\frac{1}{1-q} E\Big[ \Big( 1 +  \xi \Big)^{\frac{1}{1-q}} \ln^+ \big( 1 +  \xi  \big)     \Big]  \Big)
\end{align*}
with $ \mathcal{Y}_u = \Big( 1 + (1-q) \Big(\frac{1}{q-1}+X+ \kappa +A(0,u)\Big)\Big)^{\frac{1}{1-q}} = \Big( 1 + \xi \Big)^{\frac{1}{1-q}}$ $(\xi >0)$.\\
Now, we see that
\begin{align*}
E\Big[ \Big( 1 +  \xi \Big)^{\frac{1}{1-q}} \ln^+ \big( 1 +  \xi  \big)     \Big] 
&\leq 
\begin{cases}
2^{\frac{1}{1-q}} \ln 2, &\quad \xi \in (0,1)\\
(2\xi)^{\frac{1}{1-q}} \ln 2 + (2\xi)^{\frac{1}{1-q}} \ln \xi, &\quad \xi \geq 1
\end{cases}\\
&\leq 
\begin{cases}
2^{\frac{1}{1-q}} \ln 2, &\quad \xi \in (0,1)\\
(2e)^{\frac{1}{1-q}} \ln 2 + (1+\ln2)(2\xi)^{\frac{1}{1-q}} \ln \xi, &\quad \xi \geq 1
\end{cases},
\end{align*}
thus
$$
E\big[  \sup_{t\in [0,u]}  \mathcal{Y}_t  \big]  <\infty 
\quad \Longleftrightarrow \quad
E\big[  \xi^{\frac{1}{1-q}}\ln^+ \xi  \big]  <\infty
$$
which is guaranteed by \eqref{eq: intcond2}. Since $X\in \mathfrak{D}^a_{qu}$, it is easy to see that $Z = \mathcal{Y}^{-q}\mathcal{Z} \in \mathcal{H}^2_{[0,u]}(\mathbb{R}^k)$.
\end{proof}

\vspace{5mm}
Hereafter, we consider BSDEs with drivers $g$ both Lipschitz and of quadratic of type.
Typically, in the Lipschitz case and the quadratic case \eqref{eq: gquad}, the terminal value $X\in L^2(\mathcal{F}_u)$, while for the quadratic case \eqref{eq: BSDEsing0}, it is required that $X\in \mathfrak{D}_u$. 
Note that a sufficient condition for \eqref{eq: intcond} is $X\in L^{2\delta +2}(\mathcal{F}_u)$.

\vspace{2mm}
For the well-known relationship among BSDEs, nonlinear expectations and dynamic risk measures in the Brownian setting,  we refer to~\cite{barrieu-el-karoui, EK-rav, jiang, laeven-et-al, peng97, rg} (under the standard assumptions and on the Brownian setting) and to~\cite{barrieu-el-karoui} (for the non-Lipschitz case).
In particular,
\begin{equation*} %\label{eq: rho_BSDEg}
\rho_{tu} (X)= \mathcal{E}^g \left(-X \vert \mathcal{F}_t \right)%, \quad X \in L^{2} (\mathcal{F}_u),
\end{equation*}
is a fully-dynamic risk measure.
From now on, the risk measures $\rho_{tu}$ generated by the BSDE with driver $g$ are defined on their domain $\mathfrak{D}_u$ of $\mathcal{F}_u$-measurable random variables taking values in $\mathcal{D}_y$ and satisfying the integrability conditions appropriate for the existence of a solution to the BSDE associated. Furthermore, $\rho_{tu}$ is convex if $g$ is convex in $(y,z)$. 
Observe that, by~\cite[Prop.~7.3]{EK-rav}, if $g(t,y,z)$ is convex in $(y,z)$ and decreasing in $y$, then the corresponding fully-dynamic risk measure is cash subadditive. Indeed, also the converse implication holds true (see~\cite[Prop.~4.13]{laeven-et-al}). 
Also, if $g$ does not depend on $y$, the risk measure is cash additive.

Hereafter we discuss restriction and normalization of the fully-dynamic risk measures induced by the BSDEs above considered. The quadratic case in item \textit{b)} below is covered by \cite[Thm. 4.2]{zheng-li}.

\begin{proposition} \label{prop: normal-restric} \hspace{14cm}\\
a) $(\rho_{tu})_{t,u}$ has the restriction property \eqref{eq: restriction} if and only if $g(t,y,0)=0$ $dP\times dt$-a.e., for any~$y \in \mathcal{D}_y \subseteq \mathbb{R}$.\\
b) Assume that $g$ is Lipschitz, of quadratic type \eqref{eq: gquad} with $f$ non-decreasing and continuous in $y$, or of quadratic type \eqref{eq: BSDEsing0}. Then the fully-dynamic risk measure $(\rho_{tu})_{t,u}$ is normalized if and only if $g(t,0,0)=0$ $dP\times dt$-a.e.. 
\end{proposition}
\begin{proof}
a) Assume now that the restriction property is satisfied. Then, for any $t,u,v$ such that $0\leq t \leq u \leq v \leq T$ and for any $X \in \mathfrak{D}_u$, the relations $Y_t^u=\rho_{tu} (X)=\rho_{tv} (X)=Y_t^v$ hold, where
\begin{eqnarray*}
Y_t^u=\rho_{tu} (X)&=&-X+ \int_t^u g(s,Y^u_s, Z^u_s) ds - \int_t^u Z^u_s dB_s.
\end{eqnarray*}
A similar representation holds for $Y_t^v=\rho_{tv} (X)$.
In particular, we have $Y_u^u=-X=Y_u^v$, for all $u \leq v$ and all $X \in \mathfrak{D}_u$, in view of the restriction property.
Taking the difference $Y_u^v - Y_u^u$, we obtain that
\begin{equation}\label{fv=mart}
\int_u^v g(s, Y_s^v, Z_s^v)ds = \int_u^v Z_s^v dB_s, \qquad \text{for all } 0 \leq u \leq v.
\end{equation}
Denote $M^{(v)}(u) \triangleq \int_0^u Z_s^v dB_s$ the (local) martingale in $u\in [0,v]$, and
$f^{(v)}(u) \triangleq \int_0^u g(s, Y_s^v, Z_s^v)ds$ the process of finite variation
in $u\in [0,v]$. 
Then, by~\cite[Prop.~15.2]{kallenberg}, a continuous (local) martingale and a process of finite variation can be equal only if the martingale is constant. Hence \eqref{fv=mart}  implies that, for any $t$, the martingale $M^{(v)} $, starting from $0$ is equal to $0$ $dP\times du$-a.e..
Hence $Z_u^v= 0$, for a.a. $u \in [0,v]$, $dP$-a.s..
Going back to \eqref{fv=mart}, we also have $f^{(v)}(u) = f^{(v)}(v)$ for all $u \in [0,v]$, $dP$-a.s..
Taking the derivative with respect to $u$, we have then that $g(u, Y^v_u, 0) = 0$ for a.a. $u \in [0,v]$, $dP$-a.s..
In particular, $0 = g(v, Y^v_v, 0) = g(v, -X, 0)$ for all $X\in L^2(\mathcal{F}_u)$, for a.a. $v$, $dP$-a.s.. Finally we obtain $g(v, y, 0) = 0$ for a.a. $v \in [0,T], y \in \mathbb{R}$.

The converse implication was proved for cash additive risk measures in~\cite{DNRG1}. 
Assume that $g(t,y,0)=0$ for a.a. $t \in [0,T]$ and all $y \in \mathbb{R}$. For any $X \in L^{2}(\mathcal{F}_u)$ and $v \geq u$, consider
\begin{eqnarray*}
\rho_{tu} (X)&=&-X+ \int_t^u g(s,Y^u_s, Z^u_s) ds - \int_t^u Z^u_s dB_s
\end{eqnarray*}
and similarly for $\rho_{tv} (X)$,
where $(Y^{X,u}_r,Z^{X,u}_r)$ (resp. $(Y^{X,v}_r,Z^{X,v}_r)$) denotes the solution corresponding to $\rho_{tu} (X)$ (resp. $\rho_{tv} (X)$) at time $r \leq u$. Since $(Y^{X,v}_r,Z^{X,v}_r)$ with
$$
Y^{X,v}_r= \left\{
\begin{array}{rl}
Y^{X,u}_r;& r \leq u \\
-X;& u < r \leq v
\end{array}
\right. \qquad
Z^{X,v} (r,s)= \left\{
\begin{array}{rl}
Z^{X,u}_r;& s \leq u \\
0;& u < s \leq v
\end{array}
\right.
$$
is a solution of $\rho_{tv} (X)$ when $g(t,y,0)=0$ for a.a. $t \in [0,T]$ and all $y \in \mathbb{R}$, the restriction property follows by the uniqueness of the solution.

\vspace{2mm}
\noindent
b) Here we distinguish the cases of risk measures generated by BSDEs with Lipschitz or quadratic drivers.
We start with the Lipschitz case. Assume that $\rho_{tu}(0)=0$ for any $0\leq t \leq u \leq T$.
As shown in the proof of the Converse Comparison Theorem~\cite[Thm.~4.1]{BCHMPeng} and~\cite[Lemma~2.1]{jiang}, we have
$$
g(t,y,z)=\lim _{\varepsilon \to 0^+} \frac{\rho_{t,t+\varepsilon} \left(-y-z \cdot (B_{t + \varepsilon} -B_{t}) \right)-y}{\varepsilon}
$$
with convergence in $L^p$ with $p \in [1,2)$, for any $y \in \mathbb{R}, z \in \Bbb R^d$, and a.a. $t \leq u$.
By extracting a subsequence the convergence is $P$-a.s.
for which
$
g(t,0,0)=\lim _{\varepsilon \to 0^+} \frac{\rho_{t, t+\varepsilon} \left(0 \right)}{\varepsilon} =0
$
$dP\times dt$-a.e., since $\rho_{t,t+\varepsilon}$ is normalized. \smallskip
The converse implication is immediate.\\
The quadratic case \eqref{eq: gquad} with $f$ non-decreasing and continuous in $y$ is covered by \cite[Thm. 4.2]{zheng-li}.
The same result also covers the quadratic case \eqref{eq: BSDEsing0}, since $g(t,y,z) \leq \frac{\delta}{\epsilon}\vert z\vert^2$, for $y$ in the corresponding $\mathcal{D}_y = [-\kappa+\epsilon, \infty)$.
\end{proof}

\vspace{2mm}
We stress that in the Lipschitz case $g(t,y,0)=0$ together with the convexity of $g$ in $(y,z)$ implies that $g$ is independent of $y$.
See~\cite{BCHMPeng} Remark after Lemma 4.5. In other words, fully-dynamic risk measures induced by a Lipschitz BSDE with restriction property are necessarily cash additive. We recall that fully-dynamic risk measures (two time-parameters) with restriction property coincide with the classical dynamic risk measures (one time-parameter), see \cite{DNRG1} Equation (1.1). For this reason,  
we consider non-restricted and non-normalized fully-dynamic risk measures.

%%%%%%%%%%%%%%

\vspace{-2mm}
\subsection{Cash non-additivity and h-longevity}

We now investigate under which conditions on the driver h-longevity holds. Note that we implicitly improve the corresponding result established in~\cite{DNRG1} in the cash additive case, in fact we prove a necessary and sufficient condition and not only a sufficient one.

Let $t \leq u \leq v$ and terminal condition $X$, 
then
\begin{eqnarray*}
\rho_{tu}(X)&=& -X+ \int_t^u g(s,Y_s^u,Z_s^u) ds - \int_t^u Z_s ^u dB_s
\end{eqnarray*}
and similarly for $\rho_{tv}(X)$.
Define now the processes
\begin{equation}\label{eq: ext}
\bar{Y}_s^u=\left\{
\begin{array}{rl}
Y_s^u;& s \leq u \\
-X;& u<s \leq v
\end{array}
\right. ;
\quad \bar{Z}_s^u=\left\{
\begin{array}{rl}
Z_s^u;& s \leq u \\
0;& u<s \leq v
\end{array}
\right. .
\end{equation}

\begin{theorem} \label{prop: longevity-BSDE-nca}
Assume that $g$ satisfies the standard assumptions of Lipschitzianity in $(y,z)$.
H-longevity holds if and only if $g(s,y,0) \geq 0$ for any $s \in [0,T], y \in \mathbb{R}$.
Furthermore, in this case,  for any $t,v \in [0,T]$ with $t \leq v$, the h-longevity $\kappa$ in \eqref{eq: time-value}
$$
\gamma(t,u,v,X)=E_{\widetilde{Q}_X} \Big[ e^{\int_t^v \Delta_y g(s) ds} \int_u^v g(s,-X,0) ds | \mathcal{F}_t \Big], \quad t\leq  u\leq v, \quad X\in 
L^2(\mathcal{F}_u),
$$
\noindent where $\widetilde{Q}_X$ is a probability measure on $\mathcal{Q}_{tv}$ depending on $X$ equivalent to $P$, with density
\begin{equation*}
\frac{d \widetilde{Q}_X}{dP}= \exp \left\{ - \frac 12 \int_t^v \vert\Delta_z g(s)\vert^2 ds + \int_t^v \Delta_z g(s) dB_s \right\},
\end{equation*}
with $\Delta_z g(s) = (\Delta_z^i g(s))_{i=1,...,d}$ being defined as
\begin{equation}\label{eq: deltaz}
\Delta_z^i g(s) \triangleq \frac{g(s,\bar{Y}_s^u,Z_s^v)- g(s,\bar{Y}_s^u,\bar{Z}_s^u)}{d(Z_s^{v,i}- \bar{Z}_s^{u,i})} 1_{\{Z_s^{v,i} \neq \bar{Z}_s^{u,i}\}},
\end{equation}
while
\begin{equation}\label{eq: deltay}
\Delta_y g(s) \triangleq \frac{g(s,Y_s^v,Z_s^v)- g(s,\bar{Y}_s^u, Z_s^v)}{Y_s^v- \bar{Y}_s^u} 1_{\{Y_s^v \neq \bar{Y}_s^u\}}.
\end{equation}
\end{theorem}
Probability ${\widetilde{Q}_X}$ can be interpreted as an {\it h-longevity premium measure} (see~\cite{DNRG1}).

\smallskip
\noindent
\begin{proof}
Assume that $g(s,y,0) \geq 0$ for any $s \in [0,T]$ and $y \in \mathbb{R}$. We follow similar arguments as in~\cite{DNRG1}, however adapted to the cash non-additive case.
With the same notation above, let us also consider
\begin{equation*}
 \widetilde{Y}_s= Y_s^v - \bar{Y}_s^u; \quad \widetilde{Z}_s= Z_s^v - \bar{Z}_s^u,
\end{equation*}
for any $u \leq v$. Then
\begin{eqnarray}
\widetilde{Y}_t&&=\rho_{tv}(X)-\rho_{tu}(X) \notag\\
&&=  \int_t^v \hspace{-1mm}[g(s,Y_s^v, Z_s^v)- g(s,\bar{Y}_s^u, \bar{Z}_s^u)] ds
+ \hspace{-1mm} \int_u^v \hspace{-1mm}g(s,\bar{Y}_s^u, \bar{Z}_s^u) ds
- \hspace{-1mm}\int_t^v \hspace{-1mm} [Z_s ^v- \bar{Z}_s^u] dB_s
- \hspace{-1mm}\int_u^v \hspace{-1mm} \bar{Z}_s ^u dB_s \notag \\
&&=  \int_t^v [g(s,Y_s^v, Z_s^v)- g(s,\bar{Y}_s^u, \bar{Z}_s^u)] ds - \int_t^v \widetilde{Z}_s dB_s+ \int_u^v g(s,-X,0) ds  \notag \\
&&=  \int_t^v [\Delta_y g(s) \cdot \widetilde{Y}_s + \Delta_z g(s) \cdot \widetilde{Z}_s] ds - \int_t^v \widetilde{Z}_s dB_s+ \int_u^v g(s,-X,0) ds .\label{eq: bsde-longevity-1-nca}
\end{eqnarray}
By applying the Girsanov Theorem,~\eqref{eq: bsde-longevity-1-nca} becomes
\begin{equation}
\rho_{tv}(X)-\rho_{tu}(X)
= \int_t^v \Delta_y g(s) \cdot \widetilde{Y}_s  ds- \int_t^v \widetilde{Z}_s dB^{\widetilde{Q}_X}_s+ \int_u^v g(s,-X,0) ds, \label{eq: long-linear}
\end{equation}
where $B^{\widetilde{Q}_X}_s \triangleq B_s - B_t - \int_{t}^s \Delta_z g(r) \, dr$, $s \in [t,v]$, is a $\widetilde{Q}_X$-Brownian motion.
It is then well-known (see, e.g.,~\cite[Ex.~7.2]{EK-rav}) that
\begin{eqnarray*}
\gamma(t,u,v,X)=\rho_{tv}(X)-\rho_{tu}(X)
&=&E_{\widetilde{Q}_X} \left[ e^{\int_t^v \Delta_y g(s) ds} \int_u^v g(s,-X, 0) ds \Big\vert \mathcal{F}_t\right].
\end{eqnarray*}
By assumption on $g(\cdot, y, 0)$, it follows that $\gamma(t,u,v,X)\geq 0$.

Assume now that h-longevity holds. For any $u\leq v$, proceeding as before (see \eqref{eq: long-linear}),
\begin{equation*}
\widetilde{Y}_t= \int_u^v g(s,-X,0) ds + \int_t^v \Delta_y g(s) \cdot \widetilde{Y}_s  ds- \int_t^v \widetilde{Z}_s dB^{\widetilde{Q}_X}_s,
\end{equation*}
hence, by longevity,
\begin{equation*}
\widetilde{Y}_t =E_{\widetilde{Q}_X} \left[ e^{\int_t^v \Delta_y g(s) ds} \cdot \int_u^v g(s,-X, 0) ds \Big\vert \mathcal{F}_t\right] \geq 0
\end{equation*}
for any $t  \leq v$ and $X \in L^2 (\mathcal{F}_u)$.
Set now
\begin{eqnarray*}
\eta^F_t &=& F+ \int_t^v \left[g(s,-X,0) 1_{[u,v]}+ \Delta_y g(s) \cdot \eta^F_s \right] ds- \int_t^v Z_s^{\eta} dB^{\widetilde{Q}_X}_s \\
R^F_t &=& F+ \int_t^v  \Delta_y g(s) \cdot R^F_s  ds- \int_t^v Z_s^{R} dB^{\widetilde{Q}_X}_s
\end{eqnarray*}
for $t  \leq v$ and $F \in L^2 (\mathcal{F}_u)$. Consequently,
\begin{eqnarray*}
\eta^F_t &=&E_{\widetilde{Q}_X} \left[ e^{\int_t^v \Delta_y g(s) ds} \left(F+ \int_u^v g(s,-X, 0) ds \right) \Big\vert \mathcal{F}_t\right]  \\
R^F_t &=&E_{\widetilde{Q}_X} \left[ e^{\int_t^v \Delta_y g(s) ds} \cdot F \Big\vert \mathcal{F}_t\right]
\end{eqnarray*}
and, by longevity, $\eta^F_t=\widetilde{Y}_t + R^F_t \geq R_t^F$, for any $F \in L^2 (\mathcal{F}_u)$.
By the Converse Comparison Theorem (see~\cite{BCHMPeng, jiang}), we have
\begin{eqnarray*}
&&g(s,-X,0) 1_{[u,v]}+ \Delta_y g(s) \cdot \eta^F_s \geq \Delta_y g(s) \cdot \eta^F_s 
\end{eqnarray*}
for any $\eta_t, X$ and $s \leq v$. 
Then we have
$$
g(s,-X,0) 1_{[u,v]} \geq 0,  
$$
for all $u \leq v$, $s\leq v$, and $X\in L^2(\mathcal{F}_u)$. 
Hence, $g(s,y,0) \geq 0$ for any $s \in [0,T]$ and $y \in \mathbb{R}$.
\end{proof}

\medskip
\noindent
The previous result extends~\cite[Prop. 3.3]{DNRG1}, which deals with the cash additive case where the driver $g$ does not depend on $y$.

\begin{example}
Consider the driver
\begin{equation*}
g(t,y,z)= r_t y^- +z, \quad t \in [0,T], y \in \mathbb{R}, z \in \mathbb{R},
\end{equation*}
where $r_t$ can be interpreted as a positive interest rate depending on time $t$.
It is immediate to see that $g(\cdot,y,z)$ is decreasing in $y$, convex and Lipschitz in $(y,z)$. The corresponding fully-dynamic risk measure satisfies cash subadditivity, normalization (since $g(t,0,0)=0$, see item b) in Proposition~\ref{prop: normal-restric}) and h-longevity (by $g(t,y,0) \geq 0$ for any $t,y$, see Theorem~\ref{prop: longevity-BSDE-nca}).

Instead, for $\bar{g}(t,y,z)= r_t y^- +z +1$, with $y \in \mathbb{R}$, $z \in \mathbb{R}$,
we obtain a fully-dynamic risk measure not normalized and without restriction.
\end{example}

\begin{proposition}
    Assume $g$ is quadratic either  of type \eqref{eq: gquad} with $f$ non-decreasing and continuous or of type \eqref{eq: BSDEsing}.
    H-longevity implies that $g(t, y, 0) \geq 0$ for a.a. $t\in [0,T]$ and all $y\in \mathbb{R}$.
\end{proposition}
\begin{proof}
Note that in the case \eqref{eq: gquad}, \cite[Cor. 3.7]{zheng-li} can be directly applied, as well as in the case \eqref{eq: BSDEsing}, since we have that $g(t,y,z) \leq \frac{\delta \vert z \vert^2}{\epsilon}$. Then we have
    $$
    g(t,y,0) = \lim_{\varepsilon \to 0^+} \frac{\rho_{t, t+\varepsilon}(-y)-y}{\varepsilon} \geq 0, \qquad dP\times dt \mbox{-a.s.},
    $$
where we have used h-longevity $\rho_{t,t+\varepsilon} (-y) \geq \rho_{tt}(-y)$ and $\rho_{tt}(-y)=y$. 
\end{proof}

\begin{remark}
    The argument of the proof here above could also be used for the corresponding Lipschitz case of Theorem \ref{prop: longevity-BSDE-nca}, with adjustments in line with the proof of the first part of point b) in Proposition \ref{prop: normal-restric}.
\end{remark}

\subsection{Risk measures generated by a family of BSDEs}

Now we consider general fully-dynamic risk measures induced by a family of BSDEs of type~\eqref{eq: BSDE} with drivers $\mathcal{G}=(g_u)_{u \in [0,T]}$ depending on the time horizon $u$ of $\rho_{tu}$ and convex in $(y,z)$.
To be more precise, for any $\mathcal{F}_u$-measurable risky position $X$, assume that 
\begin{equation} \label{eq: rho-from-bsde-gt}
\rho_{tu}(X)=\rho_{tu}^{\mathcal{G}}(X) \triangleq \mathcal{E}^{g_u} (-X \vert \mathcal{F}_t), \quad t \leq u.
\end{equation}
Then $(\rho^{\mathcal{G}}_{tu})_{t,u}$ satisfies monotonicity, convexity, and continuity from above/below.
As in the cash additive case (see~\cite{DNRG1}), if $g_u(v,0,0)=0$ for any $v \leq u \leq T$, then $\rho^{\mathcal{G}}_{tu} (0)=0$ for any $t \leq u$. In general, however, this \textit{does not} guarantee the restriction property.
The following result characterizes when the restriction property holds.

\begin{proposition} %\label{prop: restriction-BSDE-family-nca}
Let $g_u$ be constant in $u$ and $g_u(t,y,0)=0$ for a.a. $t \leq u $ and for all $y \in \mathbb{R}$, then the fully-dynamic risk measure $ (\rho_{tu})_{t,u}$ satisfies the restriction property.
Vice versa, let $ (\rho_{tu})_{t,u}$ be generated by a BSDE with Lipschitz drivers or with a quadratic drivers \eqref{eq: gquad} with non-decreasing continuous $f$ or of type \eqref{eq: BSDEsing}. If it has the restriction property, then $g_u$ is constant in $u$ and $g_u(t,y,0)=0$ for a.a. $t \leq u $ and for all $y \in \mathbb{R}$.
\end{proposition}

\noindent
\begin{proof}
If $g_u$ is constant in $u$, the restriction property follows directly by Prop.~\ref{prop: normal-restric} \textit{a)}.

Conversely, assume that the restriction property holds. We study separately the Lipschitz and quadratic cases. 
Assume that $g$ is Lipschitz.
Proceeding as in the proof of the Converse Comparison Theorem of~\cite[Thm.~ 4.1]{BCHMPeng} and~\cite[Lemma~ 2.1]{jiang} and by the restriction property,
\begin{equation*}
g_u(t,y,z)=\lim _{\varepsilon \to 0^+} \frac{\rho_{tu} \left(-y-z \cdot (B_{t + \varepsilon} -B_{t})  \right)-y}{\varepsilon}=\lim _{\varepsilon \to 0^+} \frac{\rho_{t,t+\varepsilon} \left(-y-z \cdot (B_{t + \varepsilon} -B_t ) \right)-y}{\varepsilon}
\end{equation*}
with convergence in $L^p$ with $p \in [1,2)$, for any $y \in \mathbb{R}, z \in \Bbb R^d$, a.a. $t \leq u$.
Taking a subsequence, we obtain that
\begin{equation*}
\frac{\rho_{tu} \left(-y-z \cdot (B_{t + \varepsilon} -B_{t }) \right)-y}{\varepsilon} \longrightarrow g_u(t,y,z), \quad
\varepsilon \to 0^+, \quad  P-a.s.
\end{equation*}
We proceed similarly for $g_v$.
By restriction, we have
\begin{eqnarray*}
g_u(t,y,z)&=&\lim _{\varepsilon \to 0^+} \frac{\rho_{tu} \left(-y-z \cdot (B_{t + \varepsilon} -B_{t }) \right)-y}{\varepsilon} \\
&= & \lim _{\varepsilon \to 0^+} \frac{\rho_{tv} \left(-y-z \cdot (B_{t + \varepsilon} -B_{t }) \right) -y}{\varepsilon}
= g_v(t,y,z)
\end{eqnarray*}
for a.a. $u\leq v$, with $P$-a.s. convergence. This proves that $g_u$ is constant in $u$.
By Prop.~\ref{prop: normal-restric} \textit{a)}, the condition $g_u(t,y,0)=0$ should hold for a.a. $t \leq u \leq T$ and $y \in \mathbb{R}$. \\
In the quadratic cases, we follow the same argument above by applying \cite[Cor. 3.7]{zheng-li} which gives directly convergence $dP\times dt$-a.e..
\end{proof}

\vspace{2mm}
Similarly to the cash additive case h-longevity is related to the monotonicity of $\mathcal{G}$. We recall that by \textit{increasing family} $\mathcal{G}=(g_u)_{u \in [0,T]}$ it is meant, for any $u \leq v$,
\begin{equation*}
g_u(s,y,z) \leq g_v(s,y,z), \quad  \mbox{for a.a. } s \in [0,u] \mbox{ and all } y \in \Bbb R, z \in \Bbb R^d.
\end{equation*}

\begin{proposition} %\label{thm: sub-tc-bsde-gt-nca}
Let $(\rho_{tu})_{t,u}$ be the fully-dynamic risk measure~\eqref{eq: rho-from-bsde-gt} generated by Lipschitz drivers.
$\mathcal{G}$ is increasing and $g_u (s,y,0)\geq 0$ for any $u \in [0,T]$, a.a. $s\in [0,u]$, and all $y\in \mathbb{R}$, if and only if  $(\rho_{tu})_{t,u}$ satisfies h-longevity.
\end{proposition}

\noindent\begin{proof}
    Assume h-longevity holds.
By the the Converse Comparison Theorem of~\cite[Thm.~ 4.1]{BCHMPeng} and~\cite[Lemma~ 2.1]{jiang}, h-longevity yields
\begin{eqnarray*}
g_v(t,y,z)&=&\lim _{\varepsilon \to 0^+} \frac{\rho_{tv} \left(-y-z \cdot (B_{t + \varepsilon} -B_{t})  \right)-y}{\varepsilon}\\
&\geq &\lim _{\varepsilon \to 0^+} \frac{\rho_{tu} \left(-y-z \cdot (B_{t + \varepsilon} -B_{t})  \right)-y}{\varepsilon} 
= g_u(t,y,z)
\end{eqnarray*}
with convergence in $L^p$ with $p \in [1,2)$, for any $y \in \mathbb{R}, z \in \Bbb R^d$, and a.a. $t \leq u$. Hence, we can see that $\mathcal{G}$ is increasing.
Similarly, using h-longevity
\begin{eqnarray*}
g_u(t,y,0)&=&\lim _{\varepsilon \to 0^+} \frac{\rho_{tu} (-y)  -y}{\varepsilon}
\geq \lim _{\varepsilon \to 0^+} \frac{\rho_{tt} (-y)  -y}{\varepsilon} = 0.
\end{eqnarray*}

Vice versa, assume $\mathcal{G}$ is increasing and $g_u (s,y,0)\geq 0$ for any $u \in [0,T]$, a.a. $s\in [0,u]$, and all $y\in \mathbb{R}$. 
We adopt the notation of the previous section, we consider the extension \eqref{eq: ext} of the solution to the BSDE with driver $g_u$ to $[0,v]$. Correspondingly, we have $ \widetilde{Y}_s= Y_s^v - \bar{Y}_s^u$ and  $\widetilde{Z}_s= Z_s^v - \bar{Z}_s^u$. Also, we consider \eqref{eq: deltaz}-\eqref{eq: deltay} for $g_v$.
For $u\leq v$ and $X\in L^2(\mathcal{F}_u)$, we can write
\begin{eqnarray*}
\rho_{tv}(X)-\rho_{tu}(X)\hspace{-5mm} &&=  \int_t^v \hspace{-1mm}[g_v(s,Y_s^v, Z_s^v)- g_u(s,\bar{Y}_s^u, \bar{Z}_s^u)] ds
+ \hspace{-1mm} \int_u^v \hspace{-1mm}g_u(s,\bar{Y}_s^u, \bar{Z}_s^u) ds \\
&&\quad - \int_t^v \hspace{-1mm} [Z_s ^v- \bar{Z}_s^u] dB_s
- \hspace{-1mm}\int_u^v \hspace{-1mm} \bar{Z}_s ^u dB_s \notag \\
&&\geq \int_u^v g_u(s,-X,0)ds +  \int_t^v [\Delta_y g_v(s) \cdot \widetilde{Y}_s + \Delta_z g_v(s) \cdot \widetilde{Z}_s] ds - \int_t^v \widetilde{Z}_s dB_s, 
\end{eqnarray*}
by $g_v(s,\bar{Y}_s^u, \bar{Z}_s^u) - g_u(s,\bar{Y}_s^u, \bar{Z}_s^u) \geq 0$ thanks to $\mathcal{G}$ increasing.
The right-hand side here above is a linear BSDE with unique solution and 
$$
\widetilde{Y}_t= E_{\widetilde{Q}_X} \Big[ e^{\int_t^v \Delta_y g_v(s) ds} \int_u^v g_u(s,-X,0) ds | \mathcal{F}_t \Big],
$$
where
\begin{equation*}
\frac{d \widetilde{Q}_X}{dP}= \exp \left\{ - \frac 12 \int_t^v \vert\Delta_z g_v(s)\vert^2 ds + \int_t^v \Delta_z g_v(s) dB_s \right\}.
\end{equation*}
Since $\widetilde{Y}_t \geq 0$, we conclude that $\rho_{tv}(X)-\rho_{tu}(X)\geq 0$.
\end{proof} 

\vspace{2mm}
\noindent
This result improves ~\cite[Prop.~3.10]{DNRG1}, introducing the cash non-additivity, weakening the sufficient condition, and proving the necessary side.

%%%%
\section{Q-entropic risk measures on losses} \label{Sec3}

We consider the generalized entropy studied in \cite{tsallis1, tsallis2} based on the generalization of the logarithmic and exponential functions. 
This has found recent application to pricing and risk measures in \cite{ma-tian, tian}.
Here below we study the connection of the generalized entropy with BSDEs with quadratic drivers \eqref{eq: gquad} leading to the definition of the \textit{q-entropic} risk measure on losses or \eqref{eq: BSDEsing} leading to the \textit{hq-entropic} one, incorporating horizon longevity.

We start by the specific driver
\begin{equation} \label{eq: driver - entr general}
g_q(t,y,z)= \frac{q}{2} \frac{\vert z\vert^2}{1+(1-q) y}, \quad y \geq \frac{1}{q-1} + \epsilon, \quad \mbox{for }  q\in (0,1)
\end{equation}
for some $\epsilon>0$, see \eqref{eq: gquad}.
Note that when $q=1$, we reduce to the classical BSDE associated to the entropic case.
Observe that the driver $g_q$ is convex in $(y,z)$ for $y> \frac{1}{q-1}$. 
By Theorem \ref{th: sing0}, the solution exists, is unique, and non-negative for the terminal condition $X\in \mathfrak{D}_{qu}$, which is here the set of all $\mathcal{F}_u$-measurable random variables such that
$$
X\geq  \frac{1}{q-1} + \epsilon \quad \text{and} \quad E \Big[ \big(\frac{1}{1-q} + X \big)^{\frac{1}{1-q}} \ln^+ \big(\frac{1}{1-q} + X \big) \Big] <\infty.
$$
Also by Theorem \ref{th: sing0}, we can see that the solution has representation
\begin{equation} \label{eq: q-entropy representation}
Y_t= \ln_q E  \left[\left. \exp_q(X) \right| \mathcal{F}_t \right], \quad X \in \mathfrak{D}_{qu}.
\end{equation}
%In fact, we can consider
%\begin{equation}\label{eq: calY}
%\mathcal{Y}_t = \exp_q(X) - \int_t^u \mathcal{Z}_s dB_s.
%\end{equation}
%We apply It\^o formula to $d\ln_q \mathcal{Y}_t$, we see that
%$$
%-d\ln_q \mathcal{Y}_t = \frac{q}{2} \mathcal{Y}_t^{q-1} (\mathcal{Y}_t^{-q} \mathcal{Z}_t)^2 dt -  \mathcal{Y}_t^{-q} \mathcal{Z}_t dB_t.
%$$
%Taking $ Y_t = \ln_q \mathcal{Y}_t$ and $Z_t = \mathcal{Y}_t^{-q} \mathcal{Z}_t$, we see that the equation above corresponds to
%$$
%-dY_t = \frac{q}{2} \frac{\vert Z_t\vert^2}{1+(1-q) Y_t} dt - Z_t dB_t.
%$$
%Hence, from \eqref{eq: calY} we obtain \eqref{eq: q-entropy representation}.
This representation is already presented in \cite[Thm.~4.2]{ma-tian}, under more restrictive assumptions.

\vspace{2mm}
In this work, we are interested in risk evaluation of losses exceeding a given level of severity $\beta\geq 0$, representing a buffer of acceptability. Furthermore, we consider short and long term horizons triggering the need of considering interest rate uncertainty. For this reason, we construct a fully-dynamic risk measure convex and cash subadditive exploiting the BSDE approach using \eqref{eq: driver - entr general} with solution \eqref{eq: q-entropy representation}. 
For this purpose, we define the {\it q-entropic risk measure on losses}:
\begin{equation} \label{eq: g-entropic losses}
\rho_{tu}^{q}(X) \triangleq \ln_q E  \left[\left. \exp_q((X+\beta)^- +\alpha_q) \right| \mathcal{F}_t \right], \quad X\in \mathfrak{Dom}_{q,u}
\end{equation}
with $\mathfrak{Dom}_{q,u} \triangleq \{ \mathbb{R}-\text{valued } X \: \text{s.t. } (X+\beta)^-+\alpha_q \in \mathfrak{D}_{qu}\}$ and
$\alpha_q > \frac{1}{q-1}$, which is indeed fulfilling all properties requested. 

We remark that, choosing $0< \epsilon \leq \alpha_q - \frac{1}{q-1}$, then $(X+\beta)^- +\alpha_q \geq \frac{1}{q-1}+ \epsilon$ for any real-valued $X$. From now on we do consider this case.
Also, remark that
$$
E \Big[ \vert X \vert^{\frac{2-q}{1-q}}  \Big] <\infty 
\;\; \Longrightarrow \;\;
E \Big[ \big(\frac{1}{1-q} + (X+\beta)^- + \alpha_q \big)^{\frac{1}{1-q}} \ln^+ \big(\frac{1}{1-q} + (X+\beta)^- + \alpha_q \big) \Big] <\infty ,
$$
Hence $L^{\frac{2-q}{1-q}}(\mathcal{F}_u) \subset \mathfrak{Dom}_{q,u} \subset L^{\frac{1}{1-q}}(\mathcal{F}_u)$, where the last inclusion is due to similar estimate as \eqref{eq: estimate}.
Within the family \eqref{eq: g-entropic losses}, we can choose how conservative the risk evaluation should be, by selecting the values of $\alpha_q $ and $q$.
In particular, we see that, for a fixed $q$, the higher is the value $\alpha_q$, the more conservative is the corresponding risk measure $\rho^q_{tu}(X)$. In fact,
$$
\alpha_q^1 \leq \alpha_q^2 \quad \Longrightarrow \quad 
\ln_q E  \left[\left. \exp_q((X+\beta)^- +\alpha_q^1) \right| \mathcal{F}_t \right] \leq 
\ln_q E  \left[\left. \exp_q((X+\beta)^- +\alpha_q^2) \right| \mathcal{F}_t \right].
$$
Furthermore, the sensitivity with respect to $q \in (0,1)$ also reinforces how conservative $\rho^q_{tu}(X)$ is. First, note that $\mathfrak{Dom}_{q,u}$ is decreasing with $q\in (0,1)$.
\begin{proposition}
For any $\beta \in \mathbb{R}$ and $X \in \bigcap_{q\in (0,1)}\mathfrak{Dom}_{q,u} \subseteq L^\infty(\mathcal{F}_u)$, the q-entropic risk measure on losses $\rho^q_{tu}$ is increasing in $q$ with
\begin{equation*}
E [ (X+\beta)^- + \alpha_0 \vert \mathcal{F}_t]  = \rho_{tu}^0(X)
\leq \rho_{tu}^q(X) \leq \rho_{tu}^1(X) =  \ln E  \big[\exp( (X+\beta)^- + \alpha_1 ) \vert \mathcal{F}_t\big].
\end{equation*}
for $-1\leq \alpha_0 \leq \alpha_q \leq \alpha_1$.
\end{proposition}
Namely, when considered on losses, the classical entropic risk measure $\rho_{tu}^1(X)$ (with $\alpha_1=0)$ is more conservative than any q-entropic one with $-1\leq \alpha_0 \leq \alpha_q \leq \alpha_1= 0$.

\vspace{2mm}
\noindent
\begin{proof}
It is easy to check that
\begin{equation*}
\frac{\partial g_q}{\partial q}(t,y,z)=\frac{1}{2} \frac{z^2 (1+y)}{(1+(1-q) y)^2} \geq 0 \mbox{ for any } t \in [0,T], y >-1, z \in \mathbb{R}.
\end{equation*}
The increasing monotonicity of $\rho_{tu}^q(X)$ in $q$ follows from the Comparison Theorem \cite[Thm. 3.2]{ma-yao}, since $-1\leq \alpha_0 \leq \alpha_q \leq \alpha_1$.
\end{proof}

\vspace{3mm}
We remark that the q-entropic risk measure on losses $\rho^q_{tu}$ is normalized if and only if $\alpha_q = 0$. The restriction properties always holds, see Proposition \ref{prop: normal-restric} \textit{a)}.

In line with our previous discussion, we suggest a modified version of the q-entropic risk measure on losses to capture horizon risk.
We then suggest the fully-dynamic \textit{hq-entropic risk measure on losses} 
\begin{equation}\label{eq: hq-entropic}
\rho_{tu}^{hq}(X) \triangleq 
%\mathcal{E}^{{g}^{hq}}( (X+\beta)^-+\alpha_q \vert \mathcal{F}_t)= 
\ln_q E  \left[\left. \exp_q \left((X+\beta)^- +\alpha_q + A(0,u) \right) \right| \mathcal{F}_t \right] -A(0,t), \quad X\in \mathfrak{Dom}^a_{q,u},
\end{equation}
with $A(0,x) = \int_0^x a(s) ds$
and $\mathfrak{Dom}^a_{q,u} \triangleq \{ \mathbb{R}-\text{valued } X \: \text{s.t. } (X+\beta)^-+\alpha_q \in \mathfrak{D}^a_{qu}\}$, $\alpha_q> \frac{1}{q-1}$.
The risk measure above is induced by the BSDE with driver:
\begin{equation*} %\label{eq: ghq}
{g}^{hq}(t,y,z)= \frac{q}{2} \frac{z^2}{1+(1-q) (y+A(0,t))} + a(t), \quad y \geq \frac{1}{q-1} + \epsilon
\end{equation*}
for some integrable deterministic $a(t)\geq 0$, for all $t$, and $0<\epsilon\leq \alpha_q - \frac{1}{q-1}$, see Theorem \ref{th: singa}.
It is easy to see that \eqref{eq: hq-entropic} satisfies h-longevity since $a$ is non-negative.

%%%%%%%%%%%%%%%%%%%%

\section{Generalized shortfall approach} \label{Sec5}

Keeping our intent to provide families of risk measures capturing horizon risk and cash non-additivity and motivated by the fact that the classical entropic risk measure can be regarded as a shortfall risk measure (see \cite[Ex. 4.114]{follmer-schied-book}), we study this approach and provide the necessary generalizations to achieve the goal. 

\medskip

We recall that a {\it dynamic shortfall risk measure} $\rho^{U,B}_{t}$, for $t \leq u$, is given by
\begin{equation*} 
\rho^{U,B}_{t} (X)=\essinf \{ m_t \in L^p(\mathcal{F}_t): E[\left. U(X+m_t)\right| \mathcal{F}_t] \geq B_t \}, \quad X \in L^p(\mathcal{F}_u), \quad p\in [1,\infty), 
\end{equation*}
associated to a concave non-decreasing utility function $U: [-\infty, +\infty) \longrightarrow [-\infty, +\infty)$ and a target $B_t$ $\mathcal{F}_t$-measurable random variable.
Here we set $\essinf \, \emptyset = + \infty$ and $\essinf \, L^p(\mathcal{F}_t) = - \infty$. 
Observe that whenever $X\in L^p(\mathcal{F}_u)$, then 
$$
 -\infty \leq E[U(X)\vert \mathcal{F}_t] \leq U(E[X\vert \mathcal{F}_t]) <+\infty, 
$$
by Jensen inequality. 
See \cite{follmer-schied-book} for the static case on $L^\infty$ and see \cite{doldi-fritt} for the dynamic case on $L^p$ (in the context of systemic risk).
Dynamic shortfall risk measures can be equivalently formulated in terms of acceptance sets
$$\mathcal{A}^{U,B}_{t} \triangleq \{X \in L^p(\mathcal{F}_u): E[\left. U(X)\right| \mathcal{F}_t] \geq B_t \}$$
as
\begin{equation*}
\rho^{U,B}_{t} (X)=\essinf \{ m_t \in L^p(\mathcal{F}_t): X+m_t \in \mathcal{A}^{U,B}_{t} \}, \quad X \in L^p(\mathcal{F}_u).
\end{equation*}
 
\vspace{2mm}
Note that, in the sequel, we will assume that the target $B_t$ is \textit{a priori} not related to $U$.
To tackle horizon risk we introduce the following definition.
\begin{definition}\label{def: fully shortfall}
    A fully-dynamic shortfall risk measure is the family $(\rho^{U,B}_{tu})_{t,u}$:
    \begin{equation} \label{eq: shortfall-u}
\rho^{U,B}_{tu} (X)=\essinf \{ m_t \in L^p(\mathcal{F}_t): E[\left. U_u(X+m_t)\right| \mathcal{F}_t] \geq B_{tu} \}, \quad X \in L^p(\mathcal{F}_u),
\end{equation}
    generated by the families of $\mathcal{F}_t$-measurable $(B_{tu})_{t,u}$ and of mappings $(U_u)_u $ with concave and non-decreasing 
    $U_u: [-\infty, +\infty) \longrightarrow [-\infty, +\infty)$.
\end{definition}
Here $U_u$ depends also on the time horizon while the $\mathcal{F}_t$-measurable target $B_{tu}$ both on the time horizon $u$ and on the evaluation time $t$. 
The corresponding acceptance sets are given by
\begin{equation} \label{eq: shortfall-acceptance-set}
\mathcal{A}^{U,B}_{tu} \triangleq \{X \in L^p(\mathcal{F}_u): E[\left. U_u(X)\right| \mathcal{F}_t] \geq B_{tu} \}
\end{equation}
leading to the equivalent representation
\begin{equation} \label{eq: shortfall-class-acceptance}
\rho^{U,B}_{tu} (X)=\essinf \{ m_t \in L^p(\mathcal{F}_t): X+m_t \in \mathcal{A}^{U,B}_{tu} \}, \quad X \in L^p(\mathcal{F}_u).
\end{equation}

\noindent
 Hereafter, we see an example of \eqref{eq: shortfall-u} including the horizon perspective. 

 \begin{example}[Fully-dynamic h-entropic risk measure]
 Consider $U_u(x)=1-e^{- x+\int_0^u a(s) ds}$, $x\in \mathbb{R}$, with $a$ non-negative integrable real function and 
 \begin{equation*}
 B_{tu}=1- e^{\int_0^t a(s) ds}.
 \end{equation*}
 By direct computation, Definition \ref{def: fully shortfall} gives
 \begin{equation*}
 %\label{eq:h-entropic}
 \rho_{tu} (X)= \ln \left(\frac{E [e^{-X +\int_0^u a(s) ds}\vert \mathcal{F}_t]}{1-B_{tu}} \right)= \ln E [e^{-X +\int_t^u a(s) ds}\vert \mathcal{F}_t],
 \end{equation*}
 which we recognize as the h-entropic risk measure \eqref{eq: hentropica-expl} $(b=1)$.
 \end{example}
We note that $(U_u-B_{tu})_{t,u}$ is non-increasing in $u$ in the above example. This fact holds in more generality.

\begin{proposition}
    \label{prop: longevity-shortfall-u}
If $(U_u - B_{tu})_{t,u}$ is non-increasing in $u$ for all $t$, then the h-longevity of $(\rho_{tu})_{t,u}$ holds. 
Whenever both $U_{u}$ and $B_{tu}$ are constant in $u$,  restriction holds.
\end{proposition}

\noindent
\begin{proof}
By \cite[Prop.~2.8 a)]{DNRG1}, h-longevity is equivalent to $\mathcal{A}^{U,B}_{tu} \supseteq \mathcal{A}^{U,B}_{tv} \cap L^p(\mathcal{F}_u)$. The form of the acceptance sets \eqref{eq: shortfall-acceptance-set} allows to conclude.
\end{proof}
\smallskip

The following example generalizes Value at Risk to capture h-longevity. It is well-known that Value at Risk is a shortfall type risk measure (see, e.g., \cite{follmer-schied-book, BLR}), though it is not induced by BSDEs.

\begin{example}[h-Value at Risk]
The shortfall representation of Value at Risk is hereafter, for $\alpha \in (0,1)$,
\begin{equation} \label{eq: VaR-shortfall}
VaR_{\alpha}(X)=\inf\{m \in \mathbb{R}: E[U^{\alpha}(X+m)] \geq 0\},
\end{equation}
with $B=0$ and
\begin{equation*}
U^{\alpha}(x)=(\alpha-1) 1_{\{x<0\}}+\alpha 1_{\{x \geq 0\}}, \quad x \in \mathbb{R}.
\end{equation*}
Consider now the dynamic version of \eqref{eq: VaR-shortfall} in terms of \eqref{eq: shortfall-u} by taking
\begin{equation*}
U^{\alpha}_u(x)=(\alpha_u-1) 1_{\{x<0\}}+\alpha_u 1_{\{x \geq 0\}}, \quad x \in \mathbb{R},
\end{equation*}
with a deterministic $\alpha_u \in (0,1)$ depending on the time horizon $u$.
Then, for any $X \in L^p(\mathcal{F}_u)$,
\begin{eqnarray*}
\rho_{tu} (X)&=&\essinf \{ m_t \in L^p(\mathcal{F}_t): E[\left. U^{\alpha}_u(X+m_t)\right| \mathcal{F}_t] \geq B_{tu} \} \\
&=&\essinf \{ m_t \in L^p(\mathcal{F}_t): E[\left. (\alpha_u-1) 1_{\{X+m_t<0\}}+\alpha_u 1_{\{X+m_t \geq 0\}} \right| \mathcal{F}_t] \geq B_{tu} \}\\
&=&\essinf \{ m_t \in L^p(\mathcal{F}_t): E[\left. \alpha_u- 1_{\{X+m_t<0\}} \right| \mathcal{F}_t] \geq B_{tu} \}\\
&=&\essinf \{ m_t \in L^p(\mathcal{F}_t): \alpha_u -E[\left.  1_{\{X+m_t<0\}} \right| \mathcal{F}_t] \geq B_{tu} \}\\
&=&\essinf \{ m_t \in L^p(\mathcal{F}_t): \alpha_u-1  +E[\left.  1_{\{X+m_t\geq 0\}} \right| \mathcal{F}_t] \geq B_{tu} \}\\
&=&\essinf \{ m_t \in L^p(\mathcal{F}_t): \alpha_u-1  +P(\left. X+m_t\geq 0 \right| \mathcal{F}_t) \geq B_{tu} \}.
\end{eqnarray*}
Taking $B_{tu} \equiv 0$, it follows that
\begin{equation*} 
\rho_{tu} (X)=\essinf \{ m_t \in L^p(\mathcal{F}_t): P(\left. X+m_t\geq 0 \right| \mathcal{F}_t) \geq 1-\alpha_u \} \triangleq VaR_{t,u,\alpha_u}(X),
\end{equation*}
similar to the one introduced in \cite{cheridito-stadje}. 
By Proposition \ref{prop: longevity-shortfall-u}, h-longevity is guaranteed when $\alpha_u$ is decreasing in $u$. 
\end{example}

%%%%%%%%%%%%%%%%%%%%%%%%%%%%%%%%%

\subsection{Generalized shortfall risk measures}

Inspired by shortfall risk measures related to a utility function $U$ and a target $B$ that are naturally cash additive, we aim now at introducing a generalized notion of shortfall leading to cash non-additive risk measures. We further extend the definition to include h-longevity in a second stage.

\begin{definition}
Given a concave non-decreasing (non-trivial) utility $U: [-\infty, +\infty)\longrightarrow [-\infty, +\infty)$, a deterministic target $B \in \mathbb{R}$, and a function $f:\mathbb{R}^2 \to [-\infty, +\infty)$, we define and call \emph{generalized shortfall risk measure}
\begin{equation} \label{eq: generalized shortfall}
\rho^{U,f,B}(X) \triangleq \inf\{m \in \mathbb{R}:E[U(f(X,m))]\geq B\}, \quad X \in L^{p}(\mathcal{F}_T), \quad p\in [1,\infty].
\end{equation}
\end{definition}
The function $f$ can be seen as an aggregator between the cash amount $m$ and the random variable $X$. This financial interpretation is in the same spirit of systemic risk measures (see, among others, \cite{biagini-etal-1, biagini-etal-2, doldi-fritt}) where an aggregation among the different units is needed.

\vspace{4mm}
For the arguments that come in the sequel, we recall that a cash additive risk measure can be represented in terms of an acceptance set $\mathcal{A}$ as
$$
\rho(X) = \inf \{ m \in \mathbb{R}: \: X+m \in \mathcal{A} \}.
$$
In \cite[Thm.~1]{drapeau-kupper}, it is proved that the representation of cash non-additivity quasi-convex\footnote{Quasi-convex means that $\rho(\lambda X + (1-\lambda) Y) \leq \max\{ \rho(X); \rho(Y) \}$ for any $\lambda \in [0,1]$. Any convex risk measure is also quasi-convex. See, e.g., \cite{cerreia}.} risk measures is in the form
\begin{equation}
    \label{eq: dk-acceptance}
    \rho(X) = \inf \{m\in \mathbb{R}: \: X \in \mathcal{A}^m \}
\end{equation}
by means of a monotonically increasing {\it family of acceptance sets} $(\mathcal{A}^m)_{m\in \mathbb{R}}$ 
that are convex and monotone, that is
$$
X \leq Y , \; X\in \mathcal{A}^m \quad \Longrightarrow \quad Y \in \mathcal{A}^m,
$$
and right-continuous, i.e. $\mathcal{A}^m=\bigcap_{n>m} \mathcal{A}^n$ for any $m \in \mathbb{R}$.
Naturally, we have that
\begin{equation}
    \label{eq: am}
    \mathcal{A}^m = \{ X:\: \rho (X) \leq m \}.
\end{equation}

\vspace{4mm}
The generalized shortfall risk measure \eqref{eq: generalized shortfall} can be also written in terms of the family of acceptance set $(\mathcal{A}^{U,f,B;m})_m$:
\begin{equation} \label{eq: def-generalized-shortfall}
\rho^{U,f,B}(X) = \inf\{m \in \mathbb{R}: X \in \mathcal{A}^{U,f,B;m}\}, \quad X \in L^{p}(\mathcal{F}_T),
\end{equation}
with 
\begin{equation}\label{eq: accept-generalized}
\mathcal{A}^{U,f,B;m} = \{Y \in L^{p}(\mathcal{F}_T): E[U(f(Y,m))]\geq B\}.
\end{equation}

We stress that, for $f(x,m)=x+m$, the generalized shortfall risk measure reduces to a classical shortfall, see \eqref{eq: shortfall-u}.

\begin{proposition} \label{prop: properties generalized short}
Let $f$ be non-constant, $U: [-\infty, +\infty)\longrightarrow [-\infty, +\infty)$ be a concave, non-decreasing (non-trivial) function, and let  $(\mathcal{A}^{U,f,B;m})_{m \in \mathbb{R}}$ be defined as in \eqref{eq: accept-generalized}.\medskip

\noindent 1) \emph{Properties of acceptance sets:}

\noindent 1a) If $f(y,m)$ is non-decreasing in $m\in \mathbb{R}$, then the family $(\mathcal{A}^{U,f,B;m})_{m \in \mathbb{R}}$ is non-decreasing.

\noindent 1b) If $f(y,m)$ is non-decreasing in $y\in \mathbb{R}$, then $\mathcal{A}^{U,f,B;m}$ is monotone for any $m \in \mathbb{R}$.

\noindent 1c) If $(U \circ f) (y,m)$ is concave in $y$, then $\mathcal{A}^{U,f,B;m}$ is convex for any $m \in \mathbb{R}$.

\noindent 1d) If $(U \circ f) (y,m)$ is non-decreasing and right-continuous in $m$, then $(\mathcal{A}^{U,f,B;m})_{m \in \mathbb{R}}$ is right-continuous.

\medskip

\noindent 2) \emph{Properties of the generalized shortfall risk measure:}

\noindent 2a) If $f(y,m)$ is non-decreasing in $m$ and in $y$, $(U \circ f) (y,m)$ is concave in $y$,  and $(U \circ f)(y,m)$ is right-continuous in $m$, then  $\rho^{U,f,B}$ is quasi-convex.

Furthermore, under the assumptions in 2a):

\noindent 2b) If $(U \circ f) (y,m)$ is concave in $(y,m)$, then $\rho^{U,f,B}$ is convex.

\noindent 2c) If $f(y,k)\leq f(y-m,k+m)$ for any $y,k \in \mathbb{R}$ and $m>0$, then $\rho^{U,f,B}$ is cash subadditive.
\end{proposition}

\noindent
\begin{proof}
1a), 1b), 1c) can be checked easily from the definition of $(\mathcal{A}^{U,f,B;m})_{m \in \mathbb{R}}$.

\noindent 1d) For any $m \in \mathbb{R}$, the inclusion $\mathcal{A}^m \subseteq\bigcap_{n>m} \mathcal{A}^n$ follows by 1a). It remains to prove the opposite inclusion. Non-decreasing monotonicity of $(U \circ f)(y,m)$ in $m$ implies that $E[(U \circ f)(Y,n)] \geq E[(U \circ f)(Y,m)]$ for any $n >m$. By 
right-continuity of $(U \circ f)(y,m)$ in $m$ and the Monotone Convergence Theorem, then, 
$$
\lim_{n \downarrow m} E[(U \circ f)(Y,n)]=E[(U \circ f)(Y,m)].
$$ 
Assume now that $Y \in \bigcap_{n>m} \mathcal{A}^n$, hence $E[(U \circ f)(Y,n)] \geq B$ for any $n>m$. 
By the arguments above, then, also $E[(U \circ f)(Y,m)] \geq B$, that is, $Y \in \mathcal{A}^m$.

\noindent 2a) follows from the items above and by \cite[Thm.~1]{drapeau-kupper}.

\noindent 2b) By \cite[Prop.~1(i)]{drapeau-kupper}, the convexity of $\rho^{U,f,B}$ is verified if and only if $\alpha \mathcal{A}^{U,f,B;m_1}+(1-\alpha)\mathcal{A}^{U,f,B;m_2} \subseteq \mathcal{A}^{U,f,B;\alpha m_1 + (1-\alpha) m_2}$ for any $m_1,m_2 \in \mathbb{R}$ and $\alpha \in [0,1]$.
Consider $(U \circ f) (y,m)$ concave in $(y,m)$ and take arbitrary $Y_1 \in \mathcal{A}^{U,f,B;m_1}, Y_2 \in \mathcal{A}^{U,f,B;m_2}$ and $\alpha \in [0,1]$. It then follows that
\begin{eqnarray*}
&&E[U(f(\alpha Y_1 + (1-\alpha) Y_2,\alpha m_1 + (1-\alpha) m_2))] \\
&\geq & \alpha E[U(f(Y_1 ,m_1 ))] + (1-\alpha) E[U(f(Y_2 ,m_2 ))] \geq B,
\end{eqnarray*}
hence also $\alpha Y_1 + (1-\alpha) YX_2 \in \mathcal{A}^{U,f,B;\alpha m_1 + (1-\alpha) m_2}$. Convexity of $\rho^{U,f,B}$ then follows.\smallskip

\noindent 2c) By \cite[Prop.~3]{drapeau-kupper}, cash subadditivity holds if and only if $\mathcal{A}^{U,f,B;k} \subseteq \mathcal{A}^{U,f,B;k+m}+m$ for any $m>0$ and $k \in \mathbb{R}$.
In our case, we have that, for any $m>0$ and $k \in \mathbb{R}$,
\begin{eqnarray*}
\mathcal{A}^{U,f,B;k+m}+m &=& \{Y \in L^{p}(\mathcal{F}_T): E[U(f(Y,k+m))]\geq B\}+m \\
&=& \{\bar{Y} \in L^{p}(\mathcal{F}_T): E[U(f(\bar{Y}-m,k+m))]\geq B\} \\
&\supseteq & \{\bar{Y} \in L^{p}(\mathcal{F}_T): E[U(f(\bar{Y},k))]\geq B\}=\mathcal{A}^{U,f,B;k},
\end{eqnarray*}
where the inclusion is due to the assumption on $f$.
\end{proof}
\medskip

\begin{remark} 
It is easy to check that the assumption on $f$ in item 2c) is equivalent to requiring that the ratio of the increment of $f(y,k)$ in $y$ and in $k$ satisfies the following:
\begin{equation*}
\frac{\Delta_k f(y,k)}{\Delta k} \geq \frac{\Delta_y f(y,k)}{\Delta y} \quad \mbox{ for any }\Delta k=\Delta y=m >0,
\end{equation*}
where $\Delta_k f(y,k) \triangleq f(y,k + \Delta k)$ and similarly for $\Delta_y f(y,k)$.
\end{remark}

\begin{example} \hspace{10cm}\\
i) The function $f(y,m)= \beta y+m$, with $\beta \in (0,1]$, satisfies the assumption in item 2c). The case where $\beta=1$ corresponds to the classical shortfall risk measure.

\noindent
ii) The function $f(y,k)=1-\exp(-\nu y - k)$, with $\nu \in (0,1)$, is increasing in $y$ and $k$, concave in $y$ (so also $(U \circ f) (y,m)$ is concave in $y$) and such that $\frac{\partial f}{\partial y} \leq \frac{\partial f}{\partial k}$. By the arguments above, it then generates a generalized shortfall risk measure that is quasi-convex and cash subadditive.
\end{example}

\subsubsection{Dual representation of generalized shortfall risk measures}  

We recall that the cash non-additive quasi-convex $\rho(X)$ has the following dual representation
\begin{equation*}
  %  \label{eq: dk-dual}
    \rho(X) = \sup_{Q\in \mathcal{Q}} \; R(E_Q [-X],Q), \quad X\in L^p(\mathcal{F}_T),
\end{equation*}
where $\mathcal{Q}$ is the set of probability measures whose $\frac{dQ}{dP}$ belong to the dual topological space of $L^p(\mathcal{F}_T)$ and the functional $R$ is
$$
R(x,Q) = \inf \{ m \in \mathbb{R}: \; x \leq c_{min}(m,Q) \}
$$
with 
\begin{equation} \label{eq: alpha-min-qco}
 c_{min}(m,Q) = \sup_{Y \in \mathcal{A}^m} E_Q[-Y].  
\end{equation}
These arguments are detailed in \cite[pp. 38-39]{drapeau-kupper}.

We provide now a dual representation of the generalized shortfall risk measure.

\begin{proposition} \label{prop: dual repr- gener shortfall}
Let $f$ be non-constant and let $U: [-\infty, +\infty)\longrightarrow [-\infty, +\infty)$ be a concave, non-decreasing (non-trivial) function, and let $(\mathcal{A}^{U,f,B;m})_{m \in \mathbb{R}}$ be defined as in \eqref{eq: accept-generalized}.

If $f(y,m)$ is increasing in $m\in \mathbb{R}$ and in $y\in \mathbb{R}$ and $(U \circ f) (y,m)$ is concave in $y$ and right-continuous in $m$, then $\rho^{U,f,B}$ is lower semi-continuous and has the following dual representation:
\begin{equation} \label{eq: dual repr gener short}
\rho^{U,f,B}(X)= \sup_{Q \in \mathcal{Q}} R\big(E_Q[-X],Q \big), \quad  X \in L^p(\mathcal{F}_T),
\end{equation}
where, for any $x \in \mathbb{R}$ and $Q \in \mathcal{Q}$,
\begin{equation} \label{eq: R-dual repr gener short}
R(x,Q)= \inf\Big\{ m \in \mathbb{R}: \sup_{Y \in \mathcal{A}^{U,f,B;m}} E_Q[-Y] \geq x \Big\}.
\end{equation}
\end{proposition}

\noindent
\begin{proof}
\textit{Lower semi-continuity of $\rho^{U,f,B}$.} Since $(U \circ f) (y,m)$ is concave in $y$ (hence also continuous in $y$), the acceptance sets $\mathcal{A}^{U,f,B;m}$ of $\rho^{U,f,B}$ are closed for any $m \in \mathbb{R}$. Lower semi-continuity follows from \cite[Rem.~6]{drapeau-kupper}.\smallskip

\noindent \textit{Dual representation of $\rho^{U,f,B}$.} By quasi-convexity (see Prop.~\ref{prop: properties generalized short}-2a) and lower semi-continuity of $\rho^{U,f,B}$, \cite[Thm.~3]{drapeau-kupper} establishes that \eqref{eq: dual repr gener short} holds with $R(x,Q)$ being the left inverse of $c_{min}(m,Q)$, i.e.
\begin{equation*}
R(x,Q)= \inf\{ m \in \mathbb{R}:  x \leq c_{min}(m,Q)\}.  
\end{equation*}
Formulation \eqref{eq: R-dual repr gener short} follows combining \eqref{eq: alpha-min-qco} and \eqref{eq: accept-generalized}. 
\end{proof}

\begin{remark} \label{rem: bar-rho-m}
Starting from a general quasi-convex risk measure $\rho$ and its family of acceptance sets $(\mathcal{A}^m)_{m \in \mathbb{R}}$, see \eqref{eq: dk-acceptance} and \eqref{eq: am}, it is always possible to associate a family $(\bar{\rho}_m)_{m \in \mathbb{R}}$ of cash additive convex risk measures. Indeed, we can define $\bar{\rho}_m$ as the risk measure induced by the acceptance set (at level $0$) 
\begin{equation} \label{eq: accept-rho-bar-m}
\bar{\mathcal{A}}^0_m\triangleq \mathcal{A}^m  
\end{equation}
as
\begin{align}
 \bar{\rho}_m(X) & \triangleq  \inf\{ k \in \mathbb{R}: k+X \in \mathcal{A}^m\}, \nonumber \\
& =\inf\{ k \in \mathbb{R}: X \in \bar{\mathcal{A}}^k_m\},\quad X \in L^p(\mathcal{F}_T),  \label{eq: rho-bar-m} %\label{eq: rho-bar-m-2}
\end{align}
with $\bar{\mathcal{A}}^k_m\triangleq \mathcal{A}^m-k $ for $k,m \in \mathbb{R}$.
This measure is convex and cash additive (see \cite{follmer-schied-book}). Furthermore, it is easy to check that $(\bar{\mathcal{A}}^k_m)_{k \in \mathbb{R}}$ is a family of acceptance sets depending on $m \in \mathbb{R}$. Consequently, by \cite[Thm.~1]{drapeau-kupper}, the
family $(\bar{\rho}_m)_{m \in \mathbb{R}}$ is in a one-to-one correspondence with $(\bar{\mathcal{A}}^k_m)_{k, m \in \mathbb{R}}$, hence also with $(\mathcal{A}^m)_{m \in \mathbb{R}}$ and, then, with the initial $\rho$.

From \eqref{eq: accept-rho-bar-m}, we note that the minimal penalty $c_{min}^{\bar{\rho}_m}$ of $\bar{\rho}_m$ (see \cite{follmer-schied-book}) coincides with $ c_{min}(m,\cdot)$ in \eqref{eq: alpha-min-qco}. Indeed, 
\begin{equation*} 
c_{min}^{\bar{\rho}_m} (Q) \triangleq \sup_{Y \in \bar{\mathcal{A}}^0_m} E_Q[-Y]= c_{min}(m,Q), \quad  Q \in \mathcal{Q}, m \in \mathbb{R}.
\end{equation*}
\end{remark}

From the previous remark, we can deduce that a generalized shortfall risk measure can be associated one-to-one to a family of cash additive convex risk measures. That is,
\begin{equation*}
\rho^{U,f,B} \quad \longleftrightarrow \quad (\mathcal{A}^{U,f,B;m})_{m \in \mathbb{R}} \quad \longleftrightarrow \quad (\bar{\rho}_m)_{m \in \mathbb{R}}
\end{equation*}
when $(U \circ f)(y,m)$ is non-decreasing and righ-continuous in $m$.

\begin{proposition}
If $\rho$ is a quasi-convex and cash subadditive risk measure, then the family $(\bar{\rho}_m)_{m \in \mathbb{R}}$ is decreasing in $m$ and satisfies 
\begin{equation*} 
\bar{\rho}_{m+\Delta}(X) \leq \bar{\rho}_m(X) -\Delta \quad \mbox{ for any } m \in \mathbb{R}, \Delta >0, X \in L^p(\mathcal{F}_T).
\end{equation*}
Furthermore, 
\begin{equation*} 
c_{min}(m,Q) \leq c_{min}(m+h,Q)-h, \quad \mbox{ for any } Q \in \mathcal{Q}, m \in \mathbb{R}, h>0.
\end{equation*}
\end{proposition}

\noindent
\begin{proof}   
By \cite[Prop.~3]{drapeau-kupper}, if $\rho$ is cash subadditive, then its family of acceptance sets satisfies
\begin{equation} \label{eq: A-csa}
\mathcal{A}^{m+h} +h \supseteq  \mathcal{A}^m, \quad \mbox{ for any }  m \in \mathbb{R}, h >0.
\end{equation}
Consequently, for any $X \in L^p(\mathcal{F}_T)$, $m \in \mathbb{R}$, and $h>0$,
\begin{align*}
\bar{\rho}_m(X)&= \inf\{ k \in \mathbb{R}: k+X \in \mathcal{A}^m\} \\
&\geq \inf\{ k \in \mathbb{R}: k+X-h \in \mathcal{A}^{m+h}\}\\
&= \inf\{ k' \in \mathbb{R}: k'+X \in \mathcal{A}^{m+h}\}+h\\
&= \bar{\rho}_{m+h}(X)+h \geq \bar{\rho}_{m+h}(X).
\end{align*}
Furthermore, by \eqref{eq: A-csa} it holds that 
\begin{align*}
c_{min}(m,Q) &= \sup_{Y \in \mathcal{A}^m} E_Q[-Y] \\
&\leq \sup_{Y \in  \mathcal{A}^{m+h}+h} E_Q[-Y] \\
&= \sup_{Y' \in  \mathcal{A}^{m+h}} E_Q[-Y']-h \\
&=c_{min}(m+h,Q)-h,
\end{align*}
for any $Q \in \mathcal{Q}$, $m \in \mathbb{R}$, $h>0$.
\end{proof}

\subsection{H-generalized shortfall}

Starting from the definition of generalized shortfall risk measures given in \eqref{eq: generalized shortfall}, we now extend such a definition to the dynamic setting as below.

\begin{definition} 
Given a family $(f_{tu})_{0 \leq t \leq u \leq T}$ of non-constant $\mathcal{F}_t$-measurable functions $f_{tu}:\Omega \times\mathbb{R}^2 \to [-\infty, +\infty)$, a family $(U_u)_{u \in [0,T]}$ of concave non-decreasing (non-trivial) utility functions $U_u: [-\infty, +\infty)\longrightarrow \mathbb{[-\infty, +\infty)}$, and a family of $\mathcal{F}_t$-measurable targets $(B_{tu})_{0\leq t \leq u \leq T}$, we define the \emph{h-generalized shortfall risk measure} $(\rho^{U,f,B}_{tu})_{t,u}$ as
\begin{equation} \label{eq: def-dynamic-generalized-shortfall}
\rho^{U,f,B}_{tu}(X) = \essinf\{m_t \in L^p(\mathcal{F}_t): E[U_u(f_{tu}(X,m_t))\vert \mathcal{F}_t]\geq B_{tu} \}, \quad X \in L^{p}(\mathcal{F}_u).
\end{equation}
\end{definition}
The corresponding family of acceptance sets is then given by 
\begin{equation}
\label{eq: acceptance generalized shortfall dynamic}    
\mathcal{A}_{tu}^{U,f,B;m_t} = \{Y \in L^{p}(\mathcal{F}_u): E[U_u(f_{tu}(Y,m_t)) |\mathcal{F}_t]\geq B_{tu} \}.
\end{equation}
\medskip

\noindent
The following result can be easily verified similarly to Proposition \ref{prop: properties generalized short}.

\begin{proposition} \hspace{10cm}\\
\noindent a) If $(U_u \circ f_{tu}) (y,m)$ is non-decreasing in $y$ and $m$, right-continuous in $m$, and concave in $y$, then $(\rho^{U,f,B}_{tu})_{t,u}$ is a quasi-convex fully-dynamic risk measure.

Furthermore, under the assumptions in a):

\noindent b) If $(U_u \circ f_{tu}) (y,m)$ is also concave in $(y,m)$, then $(\rho^{U,f,B}_{tu})_{t,u}$ is also convex. 

\noindent c) Cash subadditivity is verified when $f_{tu}$ satisfies the further assumption in Proposition \ref{prop: properties generalized short}, item 2c).
\end{proposition}

\begin{remark}
    The dual representation of Proposition \ref{prop: dual repr- gener shortfall} can be easily extended to the fully-dynamic setting with standard arguments.
\end{remark}

In \cite{DNRG1} we have discussed and characterized horizon longevity for fully-dynamic convex risk measures \cite[Prop.~2.8]{DNRG1}. Here below we prove that, at the level of acceptance sets, a similar result holds for the quasi-convex case and we provide sufficient conditions for h-longevity to hold.

We recall that the acceptance set at level $m_t \in L^p(\mathcal{F}_t)$ of a quasi-convex $\rho_{tu}$ is 
\begin{equation}
\label{eq: accept quasi-convex dyn}
\mathcal{A}^{\rho, m_t}_{tu}\triangleq \{Y \in L^{p}(\mathcal{F}_u): \rho_{tu}(Y) \leq m_t\}.
\end{equation}

\begin{proposition}[horizon longevity]
Let $(\rho_{tu})_{t,u}$ be a fully-dynamic quasi-convex risk measure.

\noindent a) 
Horizon longevity is equivalent to $\mathcal{A}^{\rho, m_t}_{tv} \cap L^{p}(\mathcal{F}_u) \subseteq \mathcal{A}^{\rho, m_t}_{tu}$ for any $t\leq u \leq v$, $m_t \in L^{p}(\mathcal{F}_t)$.\smallskip

\noindent b) In particular, if $(\rho_{tu})_{t,u}$ is a generalized shortfall risk measure \eqref{eq: def-dynamic-generalized-shortfall} and $(U_u \circ f_{tu} - B_{tu})$ is non-increasing in $u$, then horizon longevity holds. 
\end{proposition}

\begin{proof}
a) For any $X \in \mathcal{A}^{\rho, m_t}_{tv} \cap L^{p}(\mathcal{F}_u)$, we have $\rho_{tv}(X)\leq m_t$ and, by h-longevity, also  $\rho_{tu}(X)\leq m_t$. Hence $\mathcal{A}^{\rho, m_t}_{tv} \cap L^{p}(\mathcal{F}_u) \subseteq \mathcal{A}_{tu}^{\rho, m_t}$.
Conversely, for any $X \in L^{p}(\mathcal{F}_u)$ and $t\leq u\leq v$, we have
\begin{eqnarray*}
\rho_{tv}(X)&=& \essinf \{m_t \in L^{p}(\mathcal{F}_t)| \; X \in \mathcal{A}^{ \rho, m_t}_{tv} \} \\
&\geq & \essinf \{m_t \in L^{p}(\mathcal{F}_t)| \; X \in \mathcal{A}^{\rho, m_t}_{tu} \}=\rho_{tu}(X),
\end{eqnarray*}
where the sets are given in \eqref{eq: accept quasi-convex dyn}.

\noindent Item b) is a straightforward consequence of item a) with \eqref{eq: acceptance generalized shortfall dynamic} and Proposition \ref{prop: longevity-shortfall-u}.
\end{proof}

%%%%%%%
\section{Certainty equivalent as generalized shortfall}
\label{Sec CE}

Going back to the classical entropic risk measure \eqref{eq: entropica-expl}, we know that it can be generated by a BSDE approach, a shortfall approach, and also in terms of certainty equivalent approach.
Typically, the last two are strongly related. In this section, we study the relationship between a certainty equivalent risk measure  and our newly introduced generalized shortfall.

We recall that the certainty equivalent risk measure is defined as the solution (if it exists and is unique) of 
\begin{equation} \label{eq: certainty}
\tilde{U} (-\rho^{\tilde U}(X))=E[\tilde{U} (X)], \quad  X \in L^{p}(\mathcal{F}_T),
\end{equation}
for a given concave non-decreasing (non-trivial) function $\tilde{U}: [-\infty, +\infty)\rightarrow [-\infty, +\infty)$.
See, e.g., \cite{follmer-schied-book}. Hence, when $\tilde{U}$ is strictly increasing,
\begin{equation*} %\label{eq: certainty-inverse}
\rho^{\tilde{U}}(X) \triangleq - \tilde{U}^{-1}\left( E[\tilde{U} (X)] \right), \quad  X \in L^{p}(\mathcal{F}_T).
\end{equation*}
The risk measure $\rho^{\tilde U}$ is quasi-convex, see \cite[Lemma~5.2]{cerreia} with $l(x) = - \tilde U(-x)$.
In terms of acceptance sets, $\rho^{\tilde{U}}$ can be rewritten as \eqref{eq: dk-acceptance} 
with
\begin{eqnarray*}
\mathcal{A}^{\rho^{\tilde{U}}, m} &\triangleq& \{Y \in L^{p}(\mathcal{F}_T): \rho^{\tilde{U}}(Y) \leq m\} \notag\\
&=&\{Y \in L^{p}(\mathcal{F}_T):  \tilde{U}(-\rho^{\tilde{U}}(Y) )  \geq \tilde{U}(-m)\} \notag\\
&=&\{Y \in L^{p}(\mathcal{F}_T):  E[\tilde{U} (Y)]  \geq \tilde{U}(-m)\}, 
\end{eqnarray*}
for any $m \in \mathbb{R}$. Here the last equality is due to \eqref{eq: certainty}.  Note that the acceptance sets above automatically fulfill the properties of families of acceptance sets in \cite{drapeau-kupper}.

\begin{proposition}  \label{prop: certainty-shortfall}
The certainty equivalent risk measure \eqref{eq: certainty} with $\tilde{U}$ is equal to the generalized shortfall \eqref{eq: generalized shortfall} with $(U,f,B)$ and $(U \circ f)(y,m)$ right-continuous in $m$ if and only if
\begin{equation} \label{eq: U-tilde-f}
U(f(y,m))-B= C \left(\tilde{U} (y)- \tilde{U}(-m) \right), \quad  y,m \in \mathbb{R},
\end{equation}
for some $C >0$.
Furthermore, when $U$ is invertible, this reduces to
\begin{equation}\label{eq: condition-f-cce}
f(y,m)= U^{-1}\left(C \left(\tilde{U} (y)- \tilde{U}(-m) \right)+B \right), \quad  y,m \in \mathbb{R}. 
\end{equation}
\end{proposition}

\noindent
Note that, taking $U(x)=x$ for any $x \in \mathbb{R}$, \eqref{eq: condition-f-cce} becomes
\begin{equation*}
f(y,m)= C \left(\tilde{U} (y)- \tilde{U}(-m)\right)+B, \quad \mbox{ for any } y,m \in \mathbb{R}. 
\end{equation*}

\noindent
\begin{proof}
In view of \eqref{eq: def-generalized-shortfall}-\eqref{eq: accept-generalized}, the certainty equivalent risk measure coincides with the generalized shortfall if and only if, for all $m$,
\begin{equation}\label{eq: eqv} 
\mathcal{A}^{\rho^{\tilde{U}},m}=\mathcal{A}^{U,f,B;m},
\end{equation}
namely,
\begin{equation*}
\{Y \in L^{p}(\mathcal{F}_T):  \tilde g(Y)  \geq 0 \}=\{Y \in L^{p}(\mathcal{F}_T): g(Y) \geq 0\},
\end{equation*}
where $\tilde g(Y) =E[\tilde{U} (Y)- \tilde{U}(-m)] $ and $g(Y) = E[U(f(Y,m))-B]$.\\
If \eqref{eq: U-tilde-f} holds, it is immediate to see that \eqref{eq: eqv} is true.
Conversely, fix arbitrarily $m \in \mathbb{R}$ and assume that \eqref{eq: eqv} holds. This implies that
$$
 \tilde g(Y) \geq 0 \Longleftrightarrow g(Y) \geq 0
$$
or, also,
\begin{equation} \label{eq: negative}
 \tilde g(Y) < 0 \Longleftrightarrow g(Y) < 0.
\end{equation}
Furthermore, \eqref{eq: negative} is equivalent to 
\begin{equation} \label{eq: neg-C}
\tilde g (Y) =C \cdot g(Y) \mbox{ for some } C>0.
\end{equation}
Trivially, \eqref{eq: neg-C} implies \eqref{eq: negative}. The opposite implication follows from $g(Y) \neq 0$ and
$$
\tilde g(Y)= \frac{\tilde g (Y)}{g(Y)}g(Y)= C \cdot g(Y), 
$$
where $C>0$.
The relationship \eqref{eq: neg-C} holds for all $Y$ and $m$, thus we have \eqref{eq: U-tilde-f}.
\end{proof}

\vspace{4mm} \noindent
Due to the assumptions on $U$ and on $\tilde{U}$, from \eqref{eq: U-tilde-f} it is immediate to see that
\begin{itemize}
    \item $f$ is non-decreasing in $y$ and in $m$;
    \item The concavity of $(U \circ f)$ in $y$ holds, when $\tilde{U}$ is concave;
    \item The concavity of $(U \circ f)$ in $(y,m)$ is not guaranteed in general. 
\end{itemize}

\bigskip

From a dynamic perspective, \cite{frittelli-maggis} introduced the conditional certainty equivalent risk measure. See also \cite{maggis-maran} for a further study on intertemporal preferences.
Here we provide the fully-dynamic version of the conditional certainty equivalent risk measures.

\begin{definition} 
A fully-dynamic certainty equivalent risk measure is the family $(\rho^{\tilde U_u}_{tu})_{t,u}$ with $\rho^{\tilde U_u}_{tu}$ solution (if it exists and is unique $P$-a.s.) of 
\begin{equation} \label{eq: certainty-dyn}
\tilde{U}_u (-\rho^{\tilde U_u}_{tu}(X))=E[\tilde{U}_u (X) |\mathcal{F}_t], \quad  X \in L^{p}(\mathcal{F}_u), \quad 0 \leq t \leq u,
\end{equation}
for $(\tilde{U}_u)_{u }$ of concave non-decreasing (non-trivial) functions $\tilde{U}_u:[-\infty, +\infty)\rightarrow [-\infty, +\infty)$.
Hence,
\begin{equation} \label{eq: certainty-inverse-dyn}
\rho_{tu}^{\tilde{U}_u}(X)\triangleq - \tilde{U}_u^{-1}\left( E[\tilde{U}_u (X) | \mathcal{F}_t] \right), \quad  X \in L^{p}(\mathcal{F}_u), \quad 0 \leq t \leq u,
\end{equation}
when $\tilde{U}_u(\cdot)$ is strictly increasing.
\end{definition}

As in Proposition \ref{prop: certainty-shortfall}, we can study the relation between conditional certainty equivalent and generalized shortfall risk measures.

\begin{proposition} \label{prop: certainty-shortfall-dyn}
The fully-dynamic certainty equivalent risk measure \eqref{eq: certainty-dyn} with $(\tilde{U}_u)_u$ is equal to the h-generalized shortfall \eqref{eq: def-dynamic-generalized-shortfall} with $(U_u,f_{tu},B_{tu})_{t,u}$ and $(U_u \circ f_{tu})(y,m)$ right-continuous in $m$ if and only if
\begin{equation} \label{eq: U-tilde-f-dyn}
U_u(f_{tu}(y,m))-B_{tu}= C_t \left(\tilde{U}_u (y)- \tilde{U}_u(-m)\right), \quad  y,m \in \mathbb{R}, \, 0 \leq t \leq u \leq T,
\end{equation}
for some $\mathcal{F}_t$-measurable and strictly positive $C_t$.
Furthermore, when $U_u$ is invertible, this reduces to
\begin{equation*}
f_{tu}(y,m)= U_u^{-1}\left(C_t \left(\tilde{U}_u (y)- \tilde{U}_u(-m) \right)+B_{tu} \right), \quad  y,m \in \mathbb{R}. 
\end{equation*}
\end{proposition}
\begin{proof}
    The argument is similar to Proposition \ref{prop: certainty-shortfall} with the natural adjustments to the dynamic setting. 
\end{proof}

\begin{proposition}[horizon longevity] 
Given a family $(\tilde{U}_u)_{u }$ of concave and strictly increasing functions $\tilde{U}_u(\cdot)$, the associated fully-dynamic certainty equivalent risk measure satisfies horizon longevity whenever $(\tilde{U}_u)_{u }$ is non-increasing in $u$.
\end{proposition}
\begin{proof}
From \eqref{eq: certainty-inverse-dyn}, it holds that
\begin{eqnarray*}
\rho_{tu}^{\tilde{U}_u}(X) &=& - \tilde{U}_u^{-1}\left( E[\tilde{U}_u (X) | \mathcal{F}_t] \right) \\
&\leq &- \tilde{U}_u^{-1}\left( E[\tilde{U}_v (X) | \mathcal{F}_t] \right) \\
&\leq &- \tilde{U}_v^{-1}\left( E[\tilde{U}_v (X) | \mathcal{F}_t] \right) \\
&=&\rho_{tv}^{\tilde{U}_v}(X), \quad 0 \leq t \leq u\leq v,
\end{eqnarray*}
for positions $X \in L^{p}(\mathcal{F}_u)$ such that both $\rho_{tu}^{\tilde{U}_u}(X)$ and $\rho_{tv}^{\tilde{U}_v}(X)$ are well defined.
Horizon longevity is then verified.
\end{proof}

%%%%%%%%%%%

\section{Hq-entropic on losses as a h-generalized shortfall}

As conclusion to this work, it is natural to study if the hq-entropic risk measure on losses has relationship with h-generalized shortfall and fully-dynamic certainty equivalent risk measures. We recall that
\begin{equation*}
\rho_{tu}^{hq}(X) =  \ln_q E  \left[\left. \exp_q \left((X+\beta)^- + \alpha_q + A(0,u)  \right) \right| \mathcal{F}_t \right]-A(0,t), 
\quad X\in \mathfrak{Dom}^a_{q,u}
\end{equation*}
with $A(0,u) \geq 0$ for all $u \geq 0$ and increasing in $u$, see \eqref{eq: hq-entropic}. 

\begin{theorem}
For any $t\leq u$, the hq-entropic risk measure on losses $\rho^{hq}_{tu}$ is
a h-generalized shortfall \eqref{eq: def-dynamic-generalized-shortfall} on $\mathfrak{Dom}^a_{q,u} $ with $(U_u,f_{tu},B_{tu})_{t,u}$ and $(U_u \circ f_{tu})(y,m)$ right-continuous in $m$, if and only if:

\begin{itemize}
\item for $m \geq \frac{1}{q-1}-A(0,t)$ and $y\in \mathbb{R}$,
\begin{equation}
\label{eq: hq entropic as shortfall}
 U_u(f_{tu}(y,m)) - B_{tu} = C_t\left[\exp_q \left(m+A(0,t)\right)- \exp_q \left((y+\beta)^- + \alpha_q+ A(0,u)  \right)\right] 
\end{equation}
for some $\mathcal{F}_t$-measurable and strictly positive $C_t$,

\item for $m < \frac{1}{q-1}-A(0,t)$ and  $y \in \mathbb{R}$,
\begin{equation} \label{eq: hq entropic as shortfall-2}
U_u(f_{tu}(y,m))-B_{tu} < 0 .
\end{equation}
\end{itemize}
In particular, for some target $(B_{tu})_{t,u}$ and a factor $C_t$, we can choose $U_u(x) = x$,
\begin{align}
f_{tu}(y, m) 
& =  B_{tu} + C_t\left(\exp_q \left(m+A(0,t)\right) - \exp_q \left((y+\beta)^- +\alpha_q+ A(0,u)\right)\right) \label{eq: f-hq}\\
&=  B_{tu} + C_t \Big(\exp_q (m) \Big[ 1 - \exp_q\Big( \frac{(y+\beta)^- +\alpha_q+ A(0,u) - m-A(0,t)}{1+(1-q)(m+A(0,t))}\Big) \Big]\Big), \nonumber
%\label{eq: f-hq-2}
\end{align}
for $y \in \mathbb{R}$ and $m \geq \frac{1}{q-1}-A(0,t)$, while, for $y \in \mathbb{R}$ and $m < \frac{1}{q-1}-A(0,t)$,
\begin{align*} \label{eq: f-m}
f_{tu}(y,m)=\essinf B_{tu}-\varepsilon
\end{align*}
for some $\varepsilon >0$. 
\end{theorem}

\vspace{2mm}
\noindent
\begin{proof}
Proceeding as in Prop. \ref{prop: certainty-shortfall} and \ref{prop: certainty-shortfall-dyn}, the hq-entropic risk measure on losses is a h-generalized shortfall if and only if, for all $\mathcal{F}_t$-measurable $m_t$
\begin{equation*}
\mathcal{A}_{tu}^{\rho^{hq},m_t}=\mathcal{A}_{tu}^{U,f,B;m_t},
 \mbox{ for any } 0 \leq t\leq u.
\end{equation*}
So, we need then to compare 
$$
\mathcal{A}_{tu}^{\rho^{hq}, m_t} = \big\{ Y \in \mathfrak{Dom}^a_{q,u} : \rho_{tu}^{hq}(Y) \leq m_t\big\}
$$
and
$$  
\mathcal{A}_{tu}^{U,f,B;m_t} = \{Y \in \mathfrak{Dom}^a_{q,u} ): E[U_u(f_{tu}(Y,m_t)) |\mathcal{F}_t]\geq B_{tu} \},
$$
both fulfilling the properties of families of acceptance sets as in \cite{drapeau-kupper}.
Observe that the relation $\rho_{tu}^{hq}(Y) \leq m_t$ is actually 
\begin{equation} \label{eq: A-hq}
\ln_q E  \left[\left. \exp_q \left((Y+\beta)^- +\alpha_q+ A(0,u)  \right) \right| \mathcal{F}_t \right] \leq m_t +A(0,t).
\end{equation}
By definition, for $q \in (0,1)$, $\ln_q(x)$ takes values in $\big[\frac{1}{q-1}; + \infty\big).$ 
We then consider separately the cases where $m_t +A(0,t)\geq \frac{1}{q-1}$ or not.

For $m_t +A(0,t)\geq \frac{1}{q-1}$, \eqref{eq: A-hq} becomes 
$$
E  \left[\left. \exp_q \left((Y+\beta)^- +\alpha_q+ A(0,u)  \right) \right| \mathcal{F}_t \right]\leq \exp_q (m_t + A(0,t)).
$$
Proceeding as in the proofs of Prop. \ref{prop: certainty-shortfall} and \ref{prop: certainty-shortfall-dyn}, 
$$
\mathcal{A}_{tu}^{\rho^{hq}, m_t} = \big\{ Y\::  E  \left[\left. \exp_q \left((Y+\beta)^- + \alpha_q + A(0,u)  \right) \right| \mathcal{F}_t \right]\leq \exp_q(m_t + A(0,t))\big\} =  \mathcal{A}_{tu}^{U,f,B; m_t}
$$
is then equivalent to
\begin{equation*}
 U_u(f_{tu}(y,m)) - B_{tu} = C_t\left[\exp_q \left(m+A(0,t)\right)- \exp_q \left((y+\beta)^- + \alpha_q+ A(0,u)  \right)\right] 
\end{equation*}
for all $y \in \mathbb{R}, m \geq \frac{1}{q-1}-A(0,t)$, for some $\mathcal{F}_t$-measurable and strictly positive $C_t$.
Equality \eqref{eq: f-hq} follows easily thanks also to the property $$
\frac{\exp_q(a)}{\exp_q(b)}= \exp_q \left(\frac{a-b}{1+(1-q)b} \right).
$$

Consider now the case where $m_t +A(0,t)\geq \frac{1}{q-1}$ is not satisfied $P$-a.s. In such a case, $\mathcal{A}_{tu}^{\rho^{hq}, m_t} =\emptyset $, so the same should hold for $\mathcal{A}_{tu}^{U,f,B; m_t}$, meaning that
$$
 \{Y \in \mathfrak{Dom}^a_{q,u} : E[U_u(f_{tu}(Y,m_t)) |\mathcal{F}_t]\geq B_{tu} \}= \emptyset.
$$
This is equivalent to
\begin{equation*}
U_u(f_{tu}(y,m))-B_{tu} \ngeq 0 \text{ for all } y \in \mathbb{R}, m<\frac{1}{q-1}-A(0,t).
\end{equation*}
\end{proof}

\vspace{2mm}
\noindent
Note that $(U_u\circ f_{tu} - B_{tu})_u$ is monotonically non-increasing in $u$, which is coherent with the representation of a shortfall risk measure with h-longevity, see Proposition \ref{prop: properties generalized short}.

\smallskip
\begin{remark}
   The hq-entropic risk measure on losses is not a fully-dynamic certainty equivalent risk measure.
For this, we refer to Proposition \ref{prop: certainty-shortfall-dyn} and we can see that, the two relationships 
\eqref{eq: hq entropic as shortfall}-\eqref{eq: hq entropic as shortfall-2} and \eqref{eq: U-tilde-f-dyn} cannot be the same for all $y,m \in \mathbb{R}$ even in the cases $\beta = 0$, $\alpha_q =0$, and $A \equiv 0$. 
\end{remark}

\smallskip

Hereafter, we discuss the dual representation of the hq-entropic risk measure on losses and the associated family $(\bar{\rho}_m)_{m \in \mathbb{R}}$ defined \eqref{eq: rho-bar-m}. 

\begin{remark}
Consider now, for simplicity, the static case where $t=0$ and $u=T$.
The $(U,f,B)$ with $U(x) = x$ and $f$, $B$ as in \eqref{eq: f-hq} satisfy the assumptions of Proposition \ref{prop: dual repr- gener shortfall}, hence the dual representation \eqref{eq: dual repr gener short}-\eqref{eq: R-dual repr gener short} applies also to the hq-entropic risk measure on losses. 
Observe that, by \eqref{eq: f-hq} and $A(0,0)=0$, for $m \geq \frac{1}{q-1}-A(0,T)$,
\begin{align*}
 \mathcal{A}^{U,f,B;m} &= \{ Y \in \mathfrak{Dom}^a_{q,T} : E[\exp_q (m)- \exp_q \left((Y+\beta)^- + \alpha_q + A(0,T)  \right)] \geq 0 \} \notag \\
 &=\{ Y \in \mathfrak{Dom}^a_{q,T}:  E[ \exp_q \left((Y+\beta)^- + \alpha_q + A(0,T)  \right)]  \leq \exp_q (m) \} \notag \\
 &=\{ Y \in \mathfrak{Dom}^a_{q,T}:  \ln_q E[ \exp_q \left((Y+\beta)^- +\alpha_q +  A(0,T)  \right)]\leq m \}, \label{eq: A-m-hq}
\end{align*}
while $ \mathcal{A}^{U,f,B;m} = \emptyset$ for $m < \frac{1}{q-1}-A(0,t)$.
Proceeding as in Remark \ref{rem: bar-rho-m}, we can associate a family $(\bar{\rho}_m)_{m}$ of convex cash additive risk measures to the hq-entropic risk measure on losses $\rho^{hq}$. In the present case, such a family reduces, for any $X \in \mathfrak{Dom}^a_{q,T}$, to
\begin{align*} 
 \bar{\rho}_m(X)&= \inf\{ k \in \mathbb{R}: k+X \in \mathcal{A}^{U,f,B;m}\} \\
 &= \inf\{ k \in \mathbb{R}: \ln_q E[ \exp_q \left((X+ k + \beta)^- + \alpha_q + A(0,T)  \right)]\leq m\} 
\end{align*}
when $m \geq \frac{1}{q-1}-A(0,T)$. In other words, the new family of risk measures depends on hq-entropic risk measures with an ``optimized'' target instead of the simple target $\beta$, in the sense that one is minimizing over the translation parameter $k$.
\end{remark}

\vspace{4mm}
\noindent
{\bf Acknowledgments.} We thank the Referees for their careful reading and deep comments as well as Zakaria Bensaid, Bruno Dupire, Nicole El Karoui, Marco Frittelli, Len Garces,  Sebastian Jaimungal, S\'ebastien Lleo,  Silvana Pesenti, Frank Riedel, and 
Erik Schl\"ogl for their interest and fruitful discussions.

\vspace{2mm}
\noindent
{\bf Disclosure statement.} The authors report there are no competing interests to declare and that no AI tool has been used in this research and in the preparation of the manuscript.


\begin{thebibliography}{99}

\bibitem{bahlali-et-al} Bahlali K., Eddahbi M.H., and Ouknine Y. Quadratic BSDE with $L^2$-terminal data: Krylov's estimate, It\^{o}-Krylov's formula and existence results. {\it Ann. Probab.}, 2017, 45, 2377--2397.

\bibitem{bahlali-tangpi} Bahlali, K. and Tangpi, L. BSDEs driven by $| z|^ 2/y $ and applications to PDEs and decision theory. arXiv:1810.05664.

\bibitem{barrieu-el-karoui} Barrieu P., and El Karoui N. Pricing, hedging and optimally designing derivatives via minimization of risk measures. In: {\it Indifference Pricing}, edited by R. Carmona, 2009 (Princeton University Press).

\bibitem{BLR} Bellini, F., Laeven, R. J., and Rosazza Gianin, E. Robust return risk measures. {\it Math. Financ. Econ.}, 2018, 12, 5--32.

\bibitem{biagini-etal-1} Biagini, F., Fouque, J. P., Frittelli, M., and Meyer-Brandis, T. A unified approach to systemic risk measures via acceptance sets. {\it Math. Finance}, 2019, 29, 329--367.

\bibitem{biagini-etal-2} Biagini, F., Fouque, J. P., Frittelli, M., and Meyer-Brandis, T. On fairness of systemic risk measures. {\it Finance Stoch.}, 2020, 24, 513--564.

\bibitem{bion-nadal-di-nunno} Bion-Nadal J. and Di Nunno G. Fully-dynamic risk-indifference pricing and no-good-deal bounds. {\it SIAM J. Financial Math.}, 2020, 11, 620--658.

\bibitem{BCHMPeng} Briand P., Coquet F., Hu Y., M\'{e}min J., and Peng S. A converse comparison theorem for BSDEs and related properties of g-expectation. {\it Electron. Commun. Probab.}, 2000, 5, 101--117.

\bibitem{cerreia}
Cerreia‐Vioglio, S., Maccheroni, F., Marinacci, M., and Montrucchio, L. Risk measures: rationality and diversification. {\it Math. Finance}, 2011, 21, 743--774.

\bibitem{cheridito-stadje} Cheridito, P. and Stadje, M. Time-inconsistency of VaR and time-consistent alternatives. \textit{Finance Research Letters}, 2009, 6, 40--46.

\bibitem{DNRG1} Di Nunno G. and Rosazza Gianin E. Fully-dynamic risk measures: horizon risk, time-consistency, and relations with BSDEs and BSVIEs. {\it SIAM J. Financial Math.}, 2024, 15, 399--435.

\bibitem{doldi-fritt} Doldi, A. and Frittelli, M. Conditional systemic risk measures. \textit{SIAM J. Financial Math.}, 2021, 12, 1459--1507.

\bibitem{drapeau-kupper} Drapeau, S. and Kupper, M. Risk preferences and their robust representation. \textit{Math. Oper. Res.}, 2013, 38, 28--62.


\bibitem{EK-rav} El Karoui N. and Ravanelli C. Cash subadditive risk measures and interest rate ambiguity. {\it Math. Finance} 2009, 19, 561--590.

\bibitem{farkas}
Farkas W., Koch-Medina P., and Munari C. Beyond cash-additive risk measures: when changing the num\'{e}raire fails. {\it Finance Stoch.}, 2014, 18, 145--173.

\bibitem{filipovic}
Filipovic D. Optimal numeraires for risk measures. {\it Math. Finance}, 2008, 18, 333--336.

\bibitem{follmer-schied-book} F\"{o}llmer, H. and Schied, A. {\it Stochastic Finance}, third edition, 2011 (De Gruyter: Berlin).

\bibitem{frittelli-maggis} Frittelli, M. and Maggis, M. Conditional certainty equivalent. \textit{Int. J. Theor. Appl. Finance}, 2011, 14, 41--59.

\bibitem{gilat}
Gilat, D. The best bound in the $LlogL$ inequality of Hardy and Littlewood and its martingale counterpart.
{\it Proc. Amer. Math. Soc.}, 1986, 97(3), 429-–436.

\bibitem{wang-csa}
Han X., Wang Q., Wang R., and Xia J. Cash-subadditive risk measures without quasi-convexity.  Math. Oper. Res., 2026, 51, 1205--1226.

\bibitem{jiang} 
Jiang L. Convexity, translation invariance and subadditivity for g-expectations and related risk measures. {\it Ann. Appl. Probab.}, 2008, 18, 245--258.

\bibitem{kallenberg}
Kallenberg O. {\it Foundation of modern probability}, 1997 Springer.


\bibitem{kolylanski} Kobylanski M. Backward stochastic differential equations and partial differential equations with quadratic growth. {\it Ann. Probab.}, 2000, 28, 558--602.

\bibitem{laeven-et-al} Laeven R.J., Rosazza Gianin E., and Zullino M. Star-shaped and dynamic return risk measures via BSDEs. \textit{Finance Stoch.}, 2026, 30, 903–-950.

\bibitem{laeven-et-al2} Laeven R.J., Rosazza Gianin E., and Zullino M. Geometric BSDEs. ArXiv 2405.09260v4.


\bibitem{ma-tian} Ma H. and Tian D. Generalized entropic risk measures and related BSDEs. {\it Statist. Probab. Lett.}, 2021, 174, 109110.

\bibitem{ma-yao} Ma J. and Yao S. On quadratic g-evaluations/expectations
and related analysis, {\it Stoch. Anal. Appl.}, 2010, 28, 711--734. 

\bibitem{maggis-maran} Maggis, M. and Maran, A. Stochastic dynamic utilities and intertemporal preferences. {\it Math. Financ. Econ.}, 2021, 15, 611--638.

\bibitem{mastrogiacomo-rg}
Mastrogiacomo E. and Rosazza Gianin E. Time-consistency of cash-subadditive risk measures. arXiv:1512.03641.

\bibitem{peng97} Peng S. BSDE and related g-expectations. In  {\it Pitman Research Notes in Mathematics Series}, edited by N. El Karoui and L. Mazliak, 1997, 364, 141--159.



\bibitem{rg} Rosazza~Gianin E. Risk measures via g-expectations. {\it Insurance Math. Econom}, 2006, 39, 19--34.

\bibitem{tian} Tian, D. Pricing principle via Tsallis relative entropy in incomplete markets. \textit{SIAM J. Financial Math.}, 2023, 14, 250--278.

\bibitem{tsallis1}
Tsallis C. Possible generalization of Boltzmann-Gibbs statistics. {\it J. Stat. Phys.}, 1988, 52, 479--487.

\bibitem{tsallis2}
Tsallis C. {\it Introduction to nonextensive statistical mechanics: approaching a complex world}, 2009 (Springer: New York).

\bibitem{zheng-li}
Zheng, S.Q. and Li, S.M. Representation theorem for generators of quadratic BSDEs. {\it Acta Math. Appl. Sinica (English Ser.)}, 2018, 34, 622--635.

\end{thebibliography}
\end{document}